\documentclass{article}

\usepackage{makeidx}
\usepackage{latexsym}
\usepackage{amsfonts}
\usepackage{amssymb}
\usepackage{amsmath}
\usepackage{amstext}
\usepackage{amsthm}
\usepackage{mathrsfs}

\newtheorem{theorem}{Theorem}

\newtheorem{lemma}[theorem]{Lemma}

\newtheorem{corollary}[theorem]{Corollary}

\newtheorem{proposition}[theorem]{Proposition}

\begin{document}

\title{Resolute refinements of social choice correspondences\footnote{Daniela Bubboloni was supported by GNSAGA of INdAM.
} }
\author{\textbf{Daniela Bubboloni}\\
{\small {Dipartimento di Scienze per l'Economia e  l'Impresa} }\\
\vspace{-6mm}\\
{\small {Universit\`{a} degli Studi di Firenze} }\\
\vspace{-6mm}\\
{\small {via delle Pandette 9, 50127, Firenze, Italy}}\\
\vspace{-6mm}\\
{\small {e-mail: daniela.bubboloni@unifi.it}}\\
\vspace{-6mm}\\
{\small tel: +39 055 2759667} \and \textbf{Michele Gori}
 \\
{\small {Dipartimento di Scienze per l'Economia e  l'Impresa} }\\
\vspace{-6mm}\\
{\small {Universit\`{a} degli Studi di Firenze} }\\
\vspace{-6mm}\\
{\small {via delle Pandette 9, 50127, Firenze, Italy}}\\
\vspace{-6mm}\\
{\small {e-mail: michele.gori@unifi.it}}\\
\vspace{-6mm}\\
{\small tel: +39 055 2759707}}

\maketitle

\begin{abstract}
\noindent Many classical social choice correspondences are resolute only in the case of two alternatives and an odd number of individuals. Thus, in most cases, they admit several resolute refinements, each of them naturally interpreted as a tie-breaking rule, satisfying different properties. In this paper we look for classes of social choice correspondences which admit resolute refinements fulfilling suitable versions of anonymity and neutrality.
In particular, supposing that individuals and alternatives have been exogenously partitioned into subcommittees and subclasses, we find out arithmetical conditions on the sizes of subcommittees and subclasses that are necessary and sufficient for making any social choice correspondence which is efficient, anonymous with respect to subcommittees, neutral with respect to subclasses and possibly immune to the reversal bias admit a resolute refinement sharing the same properties. 
\end{abstract}

\vspace{4mm}


\vspace{2mm}


\section{Introduction}

Consider a committee having $h\ge 2$ members who have to select one or more elements within a set of $n\ge 2$ alternatives. Assume further that the procedure used to make that choice only depends on committee members' preferences on alternatives and that such preferences are expressed as linear orders. Calling preference profile any list of $h$ preferences, each of them associated with one of the individuals in the committee, the selection procedure can be represented by a social choice correspondence ({\sc scc}), that is, a function from the set of preference profiles to the set of nonempty subsets of the set of alternatives.

In the literature many {\sc scc}s have been proposed and studied. Most of them satisfy three requirements which are considered strongly desirable by social choice theorists, namely efficiency, anonymity and neutrality. Recall that a {\sc scc} is said efficient if, for every preference profile, it does not select an alternative which is unanimously beaten by another alternative; anonymous if the identities of individuals are irrelevant to determine the social outcome, that is, it selects the same social outcome for any pair of preference profiles such that we get one from the other by permuting individual names; neutral if alternatives are equally treated, that is, for every pair of preference profiles such that we get one from the other by permuting alternative names, the social outcomes associated with them coincide up to the considered permutation.

Since in many cases collective decision processes are required to select a unique alternative, an important role in social choice theory is played by resolute {\sc scc}s, namely those {\sc scc}s associating a singleton with any preference profile.
Unfortunately, resoluteness is rarely satisfied by classical {\sc scc}s. For instance, as described in Section \ref{B-C-K-rb},  the Borda, the Copeland, the Minimax and the
Kemeny {\sc scc}s  are all efficient, anonymous and neutral but they are resolute if and only if the number of alternatives is two and the number of individuals is odd.
As a consequence, if the members of a committee want to use a classical {\sc scc} to make their collective choice and a unique outcome is needed, then they also need a tie-breaking rule to apply to the alternatives selected by the chosen {\sc scc}.

The concept of tie-breaking rule can be naturally formalized in terms of refinement of a {\sc scc}.
Let $C$ and $C'$ be two {\sc scc}s. We say that $C'$ is a refinement of $C$ if, for every preference profile, the set of alternatives selected by $C'$ is a subset of the set of alternatives selected by $C$. Thus, refinements of $C$ can be thought as a way to reduce the ambiguity in the choice made by $C$. In particular, resolute refinements of $C$ eliminate any ambiguity leading to a unique winner, so that they can be identified with tie-breaking rules.
Of course, if $C$  is not resolute, then it admits more than one resolute refinement. Thus, an interesting issue to address is to understand whether it is possible to find resolute refinements of $C$ which satisfy suitable properties making them more appealing.

In this paper we focus on the properties of efficiency, anonymity and neutrality previously described as well as on the immunity to the reversal bias. Recall that a {\sc scc} is said immune to the reversal bias if it never associates the same singleton both with a preference profile and with the one obtained by it assuming a complete change in each committee member's mind about his/her own ranking of alternatives (that is, the best alternative gets the worst, the second best alternative gets the second worst, and so on). The immunity to the reversal bias, first introduced by Saari (1994), has not been widely explored yet and there are actually only a couple of papers completed devoted to it, namely Saari and Barney (2003) and Bubboloni and Gori (2016) to which we refer for a wide discussion on the significance of such a property.

It is immediate to understand that any resolute refinement of an efficient {\sc scc} is efficient. However, resolute refinements of anonymous [neutral; immune to the reversal bias] {\sc scc}s  are not generally anonymous [neutral; immune to the reversal bias].
That happens, for instance, for resolute refinements built using two standard methods to break ties. The first method, proposed by Moulin (1988), is based on a tie-breaking agenda, that is, an exogenously given ranking of the alternatives; the second one, is instead based on the preferences of one of the individuals appointed as tie-breaker.
Of course, the resolute refinements built through a tie-breaking agenda fail to be neutral while the ones built through  a tie-breaker fail to be anonymous. Note also that an interesting result due to Moulin (1983) states that the existence of an efficient, anonymous, neutral and resolute {\sc scc} is equivalent to the strong condition $\gcd(h,n!)=1$ (Theorem \ref{moulin-teo}). Thus, in most cases, given an efficient {\sc scc}, we cannot get any anonymous and neutral resolute refinement of it. As a consequence, in those cases, we can only look for {\sc scc}s satisfying weaker versions of the principles of anonymity and neutrality.

Bubboloni and Gori (2015) propose a possible way to weaken the principle of anonymity by assuming that individuals are divided into subcommittees and requiring that, within each subcommittee, individuals equally influence the final collective decision, while people in different subcommittees may have a different decision power. They also propose a weaker version of the principle of neutrality by assuming that alternatives are divided into subclasses and requiring that within each subclass alternatives are equally treated,  while alternatives in different subclasses may have a different treatment\footnote{It is worth mentioning that some weak versions of the principle of anonymity and neutrality are also introduced by Campbell and Kelly (2011, 2013).}. These versions of anonymity and neutrality are
certainly natural and actually used in many practical collective
decision processes. That happens, for instance, when a committee
has a president or when a committee 
evaluates job candidates discriminating on their gender. In the former example committee members can be thought to be
divided in two subcommittees (the president in the first, all the
others in the second) with anonymous individuals within each
of them; in the latter example alternatives can be thought to be
divided in two subclasses (the women in the first, the men in the
second) such that no alternative has an exogenous advantage with
respect to the other alternatives in the same subclass.

In this paper, we fist find out arithmetical conditions on the sizes of subcommittees and subclasses that are necessary for the existence of a  resolute {\sc scc} which is efficient, anonymous with respect to subcommittees and neutral with respect to subclasses [and immune to the reversal bias]
(Theorem \ref{super}). We then prove that the same conditions assure that any efficient {\sc scc} which is anonymous with respect to subcommittees and neutral with respect to subclasses [and immune to the reversal bias] admits a resolute refinement having the same properties (Theorem \ref{super-super}). Those results, among other things, generalize the previously mentioned theorem by Moulin.

While the proof of the first result is simple and natural, the proof of the second one, along with other interesting results, require a certain amount of work. The arguments are strongly based on the algebraic approach  developed in Bubboloni and Gori (2014, 2015) where, in the framework of social welfare functions, the notion of action of a group on a set is naturally and fruitfully used to study problems concerning anonymity and neutrality and weaker versions of them, along with reversal symmetry.
Here we adapt that algebraic reasoning to the framework of {\sc scc}s by defining a general and wide-ranging notion of consistency of a  {\sc scc} with respect to a group (Section \ref{scc}), which includes anonymity with respect to subcommittees, neutrality with respect to subclasses and immunity to the reversal  bias as particular instances.
That notion of consistency provides a unified framework which allows on the one hand to make
proofs simpler and more direct, and on the other hand to obtain very general results (Theorem \ref{general}).
It is worth noting that the algebraic approach developed in the paper also provides methods to potentially build all
the desired resolute refinements. In Sections \ref{example53} and \ref{example33} we discuss
some examples that explain how the theoretical results can be
explicitly applied.


\section{Preliminary definitions and facts}

Throughout the paper, given $A$, $B$ and $C$ sets and $f:A\to B$ and $g:B\to C$ functions, we denote by $gf$ the right-to-left composition of $f$ and $g$, that is, the function from $A$ to $C$ defined, for every $a\in A$, as $gf(a)=g(f(a))$.


\subsection{Groups and permutations}\label{plo}

All the results from group theory used in the paper can be found in Jacobson (1974, Chapter 1).
Below, we briefly recall some well known concepts that will be sufficient for a complete comprehension of the paper until the end of Section \ref{main}, where the main theorems of the paper are stated and commented. 

A finite group $G$ is a finite set endowed with a binary operation satisfying associativity, admitting neutral element $1_G$ and such that every element has inverse. Consider $g\in G$. We set $g^0=1_G$ and, for every $s\in \mathbb{N}$, we denote by $g^s$ the product of $g$ by itself $s$ times. We also denote by $|g|$ the order of $g$, that is, the minimum $s\in \mathbb{N}$ such that $g^s=1_G$.
A subset $U$ of $G$ is called a subgroup of $G$ if $U$ is closed under the operation in $G$, that is, if for every $u_1,u_2\in U,$ we have $u_1u_2\in U.$ If $U$ is a subgroup of $G$, we use the notation $U\leq G.$

Let $X$ be a nonempty finite set.  Then $\mathrm{Sym}(X)$ denotes  the group of the bijective functions from $X$ to itself, with product defined, for every $\sigma_1,\sigma_2\in\mathrm{Sym}(X)$, by $\sigma_1\sigma_2\in \mathrm{Sym}(X)$. The neutral element of $\mathrm{Sym}(X)$ is given by the identity function, denoted by $id$.
$\mathrm{Sym}(X)$ is called the symmetric group on $X$ and its elements are called permutations on $X$. 
For every $k\in\mathbb{N}$, the group $\mathrm{Sym}(\{1,\dots,k\})$ is simply denoted by $S_{k}$. The elements in $S_k$ are usually written via the standard representation through disjoint cycles. For instance, $\psi=(134)(26)\in S_6$ is the permutation defined by
\[
\psi(1)=3,\; \psi(3)=4, \;\psi(4)=1,\; \psi(2)=6,\; \psi(6)=2,\; \psi(5)=5.
\]

\subsection{Preference relations}\label{preference-relations}

\noindent From now on, let $n\in \mathbb{N}$ with $n\ge 2$ be fixed, and let $N=\{1,\ldots,n\}$ be the set of names of alternatives. 

A {\it preference relation} on $N$ is a linear order on $N$, that is, a complete, transitive and antisymmetric binary relation. The set of preference relations on $N$ is denoted by $\mathcal{L}(N)$. 
Given $q\in \mathcal{L}(N)$ and $x,y\in N$, we usually write $x\succeq_{q}y$ instead of $(x,y)\in q$, as well as $x\succ_{q}y$ instead of $(x,y)\in q$ and $x\neq y$, and we say that $x$ is  preferred to $y$ according to $q$ if $x\succ_{q} y$.

Let $q\in\mathcal{L}(N)$ be fixed. 
For every $\psi \in S_n$, we define $\psi q$ as the element of $\mathcal{L}(N)$ such that, for every $x,y\in N$, $(x,y)\in \psi q$ if and only if $(\psi^{-1}(x),\psi^{-1}(y))\in q$. Consider the {\it order reversing permutation} in $S_n$, that is, the permutation $\rho_0\in S_n$ defined, for every $r\in \{1,\ldots,n\}$, as $\rho_0(r)=n-r+1$. Obviously, we have $|\rho_0|=2$.  Note that
$\rho_0$ has exactly  one fixed point when $n$ is odd and no fixed point when $n$ is even. For instance, if $n=3$, we have $\rho_0= (13)$ and $2$ is the only fixed point; if $n=4$, we have $\rho_0= (14)(23)$ and no fixed point. Define $\Omega=\{id,\rho_0\}$, where $id\in S_n$. Note that $\Omega\leq S_n$ is a commutative group  which admits as unique subgroups $\{id\}$ and $\Omega$.
We define $q\rho_0\in\mathcal{L}(N)$ as the element in $\mathcal{L}(N)$ such that, for every $x,y\in N$,
$(x,y)\in q\rho_0$ if and only if $(y,x)\in q$;
$q\; id=q$. 
Note that, by definition, for every $x,y\in N$ and $\psi\in S_n$, we have that
$x\succ_{q}y$ if and only if $\psi(x)\succ_{\psi q}\psi(y)$; $x\succ_{q} y$ if and only if $y\succ_{q\rho_0} x$.

The function $\mathrm{rank}_{q}:N\to \{1,\ldots,n\}$ is defined, for every $x\in N$, by
\[
\mathrm{rank}_q(x)=|\{y\in N: y\succeq_q x\}|.
\]
 Such a function is called the rank of $x\in N$ in $q$ and is bijective. Note that, for every $\psi \in S_n$,
$\mathrm{rank}_{\psi q}(x)=\mathrm{rank}_{q}(\psi^{-1}(x))$  and
$\mathrm{rank}_{q\rho_0}(x)=\rho_0(\mathrm{rank}_{q}(x)).$

Consider now the set of vectors with $n$ distinct components in $N$ given by
\[
\mathcal{V}(N)=\left\{(x_r)_{r=1}^n\in N^n: x_{r_1}=x_{r_2}\Rightarrow r_1=r_2\right\},
\]
and think each vector $(x_r)_{r=1}^n\in \mathcal{V}(N)$ as a column vector, that is,
 \[
(x_r)_{r=1}^n=\begin{bmatrix}
x_1\\ \vdots \\ x_n
\end{bmatrix}=[x_1,\ldots,x_n]^T.
\]
The function $f_1:\mathcal{V}(N)\to \mathcal{L}(N)$ associating with  $(x_r)_{r=1}^n\in\mathcal{V}(N)$ the preference relation
\[
\{(x_{r_1},x_{r_2})\in N\times N: r_1,r_2\in \{1,\ldots,n\}, r_1\le r_2\},
\]
and the function
$f_2:S_n\to \mathcal{L}(N)$ associating with $\sigma\in S_n$  the preference relation
\[
\{(\sigma(r_1),\sigma(r_2))\in N\times N: r_1,r_2\in \{1,\ldots,n\}, r_1\le r_2\}
\]
are bijective, so that, in particular, $|S_n|=|\mathcal{V}(N)|=|\mathcal{L}(N)|=n!$.
Note that
\[
f_1^{-1}(q)=\{(x_r)_{r=1}^n\in\mathcal{V}(N): \forall r\in\{1,\ldots,n\}, \mathrm{rank}_q(x_r)=r \},
\]
\[
 f_2^{-1}(q)=\{\sigma\in S_n: \forall r\in\{1,\ldots,n\}, \mathrm{rank}_q(\sigma(r))=r \}.
\]
Note also that, for every $\psi\in S_n$ and $\rho\in\Omega$,  if $f_1^{-1}(q)=[x_1,\ldots,x_n]^T$, then
\[
f_1^{-1}(\psi q)=[\psi(x_1),\ldots,\psi(x_n)]^T,\quad \mbox{and}\quad f_1^{-1}(q\rho)=[x_{\rho(1)},\ldots,x_{\rho(n)}]^T;
\]
if $f_2^{-1}(q)=\sigma$, then
\[
f_2^{-1}(\psi q)=\psi\sigma,\quad \mbox{and}\quad f_2^{-1}(q\rho)=\sigma \rho.
\]
Thus, by the functions $f_1$ and $f_2$ we are allowed to identify the preference relation $q$ both with the vector $f_1^{-1}(q)$ and with the permutation  $f_2^{-1}(q)$, and to naturally interpret the products $\psi q$ and $q\rho$ in $\mathcal{V}(N)$ and in $S_n$.
For instance, if $n=4$ and
\[
q=\left\{(4,2),(2,1),(1,3),(4,1),(4,3),(2,3),(4,4),(2,2), (1,1),(3,3)\right\}\in\mathcal{L}(\{1,2,3,4\}),
\]
then $q$ is identified with both
$
f_1^{-1}(q)=[4,2,1,3]^T \in \mathcal{V}(\{1,2,3,4\})
$ and $
f_2^{-1}(q)=(143)\in S_4,
$
so that 4 has rank 1, 2 has rank 2, 1 has rank 3, and 3 has rank 4 in $q$. Thus, if
$\psi=(342)\in S_4$, then we can write
\[
\psi q=(342)[4,2,1,3]^T=[2,3,1,4]^T \quad \mbox{and}\quad q\rho_0=[4,2,1,3]^T(14)(23)=[3,1,2,4]^T,
\]
as well as
\[
\psi q=(342)(143)=(123) \quad \mbox{and}\quad q\rho_0=(143)(14)(23)=(132).
\]
On the one hand,  identifying preference relations with vectors makes computations easy and intuitive.
On the other hand, identifying preference relations with permutations allows to transfer the group properties of $S_n$
to the products between preference relations and permutations. In particular, by associativity and cancellation laws, for every $\psi_1,\psi_2\in S_n$ and $\rho_1,\rho_2\in\{id, \rho_0\}$, we have that $\psi_1 q=\psi_2 q$ if and only if $\psi_1=\psi_2$;  $q\rho_1=q\rho_2$ if and only if $\rho_1=\rho_2$; $(\psi_2\psi_1)q=\psi_2(\psi_1 q)$; $q(\rho_1\rho_2)=(q\rho_1)\rho_2$; $(\psi_1 q)\rho_1=\psi_1(q \rho_1)$.
For every $\psi\in S_n$, $\rho\in \Omega$ and $X\subseteq \mathcal{L}(N)$, we define $\psi X \rho=\{\psi q \rho\in \mathcal{L}(N): q\in X\}$. 
Note that, for every $\psi\in S_n$, $\rho\in \Omega$ and $X\subseteq Y\subseteq \mathcal{L}(N)$, $\psi X \rho\subseteq \psi Y \rho$ so that, in particular, $\psi\mathcal{L}(N)\rho=\mathcal{L}(N)$.

Given now $\psi\in S_n$ and $\rho\in\{id,\rho_0\}$, we finally emphasize that the above discussion makes the products  $\psi q$ and $q\rho$ have interesting interpretations. Indeed, if $q$ represents the preferences of a certain individual, then $\psi q$  represents the preferences that the individual would have if, for every $x\in N$, alternative $x$ were called $\psi(x)$; $q\rho$ represents the preferences that the individual would have if, for every $r\in \{1,\ldots,n\}$, the alternative whose rank is $r$ is moved to rank $\rho(r)$.

\subsection{Preference profiles}\label{model}

\noindent From now on, let $h\in \mathbb{N}$ with $h\ge 2$ be fixed, and let $H=\{1,\ldots,h\}$ be the set of names of individuals.
A {\it preference  profile} is an element of $\mathcal{L}(N)^h$. The set $\mathcal{L}(N)^h$ is denoted by $\mathcal{P}$.
If $p\in\mathcal{P}$ and $i\in H$, the $i$-th component of $p$ is denoted by $p_i$ and represents the preferences of individual $i$.
 Any $p\in\mathcal{P}$ can be identified with the $n\times h$ matrix whose $i$-th column is the column vector representing $p_i$ for all $i\in H$.

Let us consider the set $G=S_h\times S_n\times \Omega$. Note that $G$ is a group through component-wise multiplication, that is, defining for every $(\varphi_1,\psi_1,\rho_1)\in G$ and $(\varphi_2,\psi_2,\rho_2)\in G$, $(\varphi_1,\psi_1,\rho_1)(\varphi_2,\psi_2,\rho_2)=
(\varphi_1\varphi_2,\psi_1\psi_2,\rho_1\rho_2).$
For every $(\varphi,\psi,\rho)\in G$ and $p\in \mathcal{P}$, define $p^{(\varphi,\psi,\rho)} \in \mathcal{P}$ as the preference profile such that, for every $i\in H$,
\[
(p^{(\varphi,\psi,\rho)})_i=\psi p_{\varphi^{-1}(i)}\rho.
\]
Thus, the preference profile $p^{(\varphi,\psi,\rho)}$ is obtained by $p$ according to the following rules (to be applied in any order): for every $i\in H$, individual $i$ is renamed $\varphi(i)$; for every $x\in N$, alternative $x$ is renamed $\psi(x)$; for every $r\in\{1,\ldots,n\}$, alternatives whose rank is $r$ are moved to rank $\rho(r)$.

For instance, if $n=3$, $h=7$ and
\[
p=
\left[
\begin{array}{ccccccccc}
3 & 1 & 2 & 3 &2&1&1\\
2 & 2 & 1 & 2 &3&3&3\\
1 & 3 & 3 & 1 &1&2&2
\end{array}
\right],\quad \varphi=(134)(25),\quad \psi=(12),\quad \rho=\rho_0=(13),
\]
then we have
\[
p^{(\varphi,id,id)}=
\left[
\begin{array}{ccccccc}
3 & 2 &3  &2  &1& 1& 1\\
2 & 3 & 2 & 1 & 2& 3& 3\\
1 & 1 & 1 & 3 & 3& 2& 2
\end{array}
\right],\quad p^{(id,\psi,id)}=
\left[
\begin{array}{ccccccccc}
 3& 2 & 1 & 3 &1&2& 2\\
 1& 1 &2  & 1 &3& 3& 3\\
 2& 3 & 3 & 2 &2& 1& 1\\
 \end{array}
\right],
\]
\[
p^{(id,id,\rho_0)}=
\left[
\begin{array}{ccccccc}
1 & 3 & 3 & 1 &1&2& 2\\
2 & 2 & 1 & 2 &3& 3& 3\\
3 & 1 & 2 & 3 & 2& 1& 1 \\
 \end{array}
\right],\quad
p^{(\varphi,\psi,\rho_0)}=
\left[
\begin{array}{ccccccc}
2 & 2 & 2 & 3 &3&1&1\\
1 & 3 & 1 & 2 &1&3&3\\
3 & 1 & 3 & 1 & 2 &2&2\\
 \end{array}
\right].
\]
As it is easy to verify, if $n=2$, then $p^{(id,\rho_0,id)}=p^{(id,id,\rho_0)}$ for all $p\in\mathcal{P}$; if $n\ge 3$, then
there do not exist $\varphi\in S_h$ and $\psi\in S_n$ such that, for every $p\in\mathcal{P}$, $p^{(\varphi,\psi,id)}=p^{(id,id,\rho_0)}$. In other words,
top-down reversing preference profiles cannot be reduced, in general, to a change in individuals and alternatives names.
In what follows, we write the $i$-th component  $p^{(\varphi,\psi,\rho)}$ simply as $p^{(\varphi,\psi,\rho)}_i,$ instead of $(p^{(\varphi,\psi,\rho)})_i$.

\subsection{Social choice correspondences}\label{scc}

A {\it social choice correspondence} ({\sc scc}) is a function from $\mathcal{P}$ to the set of the nonempty subsets of $N$.
The set of {\sc scc}s is denoted by $\mathfrak{C}$.
Given  $C\in\mathfrak{C}$, we say that $C$ is {\it resolute} if, for every $p\in\mathcal{P}$, $|C(p)|=1$.
We say that $C'\in \mathfrak{C}$ is a {\it refinement} of $C$ if, for every $p\in\mathcal{P}$, $C'(p)\subseteq C(p)$. Note that $C$ admits a unique resolute refinement if and only if $C$ is resolute.

Consider now $Par\in\mathfrak{C}$ defined, for every $p\in\mathcal{P}$, as
\begin{equation}\label{pareto-eq}
Par(p)=\{x\in N: \forall \,y\in N\setminus\{x\}, \; \exists \,i\in H \mbox{ such that } x\succ_{p_i} y\}.
\end{equation}
$Par$ is called the Pareto {\sc scc}. Note that $Par(p)$ contains the alternatives ranked first in at least one $p_i.$ In particular, $Par(p)\neq\varnothing$.

We say that $C\in\mathfrak{C}$ is {\it efficient} if $C$ is a refinement of $Par$.
Thus, an efficient {\sc scc} never selects an alternative which is unanimously beaten by another one.
Of course, every refinement of an efficient {\sc scc} is efficient.

Given a partition\footnote{A partition of a nonempty set $X$ is a set of nonempty pairwise disjoint subsets of $X$ whose union is $X$.} $Y=\{Y_j\}_{j=1}^s$ of $H$, where $s\in\mathbb{N}$, we define the set
\[
V(Y)=\left\{\varphi\in S_h : \varphi(Y_j)=Y_j\hbox{ for all } j\in\{1,\dots,s\} \right\},
\]
and we say that $C$ is anonymous with respect to $Y$, briefly {\it $Y$-anonymous}, if, for every $p\in\mathcal{P}$ and  $\varphi\in V(Y)$, we have
\[
C(p^{(\varphi,id,id)})=C(p).
\]
Thus, interpreting the elements of $Y$ as subcommittees, we have that $Y$-anonymous {\sc scc}s attribute the same decision power to individuals in the same subcommittee. 
Note that $V(Y)\leq S_h$.

Given now a partition $Z=\{Z_k\}_{k=1}^t$ of $N$, where $ t\in\mathbb{N}$, we define the set
\[
W(Z)=\left\{\psi\in S_n :  \psi(Z_k)=Z_k \hbox{ for all }  k\in\{1,\dots,t\} \right\},
\]
and we say that $C$ is neutral with respect to $Z$, briefly {\it $Z$-neutral}, if, for every $p\in\mathcal{P}$ and  $\psi\in W(Z)$, we have\footnote{In order to simplify the notation, given $\psi\in S_n$ and $X\subseteq N$, we write $\psi X$ instead of $\psi(X)$.}
\[
C(p^{(id,\psi,id)})=\psi C(p).
\]
Thus, interpreting the elements of $Z$  as subclasses, we have that $Z$-neutral {\sc scc}s cannot distinguish among alternatives in the same subclass. Note that $W(Z)\leq S_n.$

Of course, if $Y=\{H\}$ [$Z=\{N\}$], then $V(H)=S_h$ [$W(Z)=S_n$] so that the classical requirement of anonymity [neutrality] corresponds to the one of $Y$-anonymity [$Z$-neutrality]. Moreover, $C$ is anonymous [neutral] if and only if $C$ is  $Y$-anonymous [$Z$-neutral] for all partition $Y$ of $H$ [$Z$ of $N$].

Following Bubboloni and Gori (2016), we finally say that $C$ is {\it immune to the reversal bias} if, for every $p\in\mathcal{P}$ with $|C(p)|=1$, we have
\[
C(p^{(id,id,\rho_0)})\neq C(p).
\]
In other words, a  {\sc scc} is immune to the reversal bias if it never associates the same unique winner both with a preference profile $p$ and with the preference profile obtained by $p$ reversing every individual preference.

\section{Analysis of some classical {\sc scc}s}
\label{B-C-K-rb}

The Pareto ($Par$), the Borda ($Bor$), the Copeland ($Cop$), the Minimax ($Min$) and the Kemeny ($Kem$) {\sc scc}s are classical {\sc scc}s deeply studied in the literature. Recall that $Par$ is defined in \eqref{pareto-eq} while, for every $p\in\mathcal{P}$, we have that\footnote{With Borda {\sc scc} we mean the well-known Borda count. All the other definitions are taken from Fishburn (1977).}
\[
\begin{array}{l}
Bor(p)=\underset{x\in N}{\mathrm{argmax}} \displaystyle{\sum_{i=1}^h} \left(n-\mathrm{rank}_{p_i}(x)\right),\\
\\
Cop(p)=\underset{x\in N}{\mathrm{argmax}}\, \left(\left|\left\{y\in N: w_p(x,y)\ge \left\lceil \frac{h+1}{2}\right\rceil\right\}\right|-\left|\left\{y\in N: w_p(y,x)\ge \left\lceil \frac{h+1}{2}\right\rceil\right\}\right|\right),\\
\\
Min(p)=\underset{x\in N}{\mathrm{argmax}} \underset{y\in N\setminus\{x\}}{\min} w_p(x,y),\\
\\
Kem(p)=\Big\{x\in N: \exists\, q^*\in \underset{q\in\mathcal{L}(N)}{\mathrm{argmax}}\,
k_p(q) \mbox{ with }\mathrm{rank}_{q^*}(x)=1\Big\}.
\end{array}
\]
where, for every $p\in \mathcal{P}$, $x,y\in N$ and $q\in\mathcal{L}(N)$, we have set
\[
w_p(x,y)=|\{i\in H:  x\succ_{p_i} y \}|,
\quad
k_p(q)=\sum_{x\succ_q y}w_p(x,y).
\]
It is well-known that the following proposition holds true.
\begin{proposition}\label{ean}
$Par$, $Bor$, $Cop$, $Min$ and $Kem$ are efficient, anonymous and neutral.
\end{proposition}
In particular, all the considered {\sc scc}s are anonymous with respect to any partition of $H$ and
neutral with respect to any partition of $N$. In the next propositions, we find out conditions on the number of individuals and alternatives which are necessary and sufficient to make those {\sc scc}s immune to the reversal bias and resolute.

\begin{proposition}\label{others}
$Par$, $Bor$, $Cop$ and $Kem$ are immune to the reversal bias.
\end{proposition}

\begin{proof}
That fact for the Borda and Copeland {\sc scc}s is proved in Bubboloni and Gori (2016, Proposition 3). We are thus left with considering  $C\in \{Par,Kem\}$ and showing that if, for some $p\in \mathcal{P}$ and $x,y\in N,$ we have $C(p)=\{x\}$ and $C(p^{ (id,id,\rho_0)})=\{y\}$, then $x\neq y.$
In what follows, for shortness, for every $p\in \mathcal{P}$, we will write $p^{\rho_0}$ instead of $p^{(id,id,\rho_0)}.$

Assume that  $Par(p)=\{x\}$ and $Par(p^{ \rho_0})=\{y\}$ for some $p\in\mathcal{P}$ and $x,y\in N$. Then, in particular,  we have that $\mathrm{rank}_{p_1}(x)=1$ and $\mathrm{rank}_{p_1^{\rho_0}}(y)=1$. Thus, $\mathrm{rank}_{p_1}(y)=n\neq 1=\mathrm{rank}_{p_1}(x),$ which says $x\neq y$.

Consider now $Kem$. We start defining, for every $p\in\mathcal{P}$, the nonempty subset of $\mathcal{L}(N)$
\[
\mathcal{K}(p)=\underset{q\in\mathcal{L}(N)}{\mathrm{argmax}}\,
k_p(q)
\]
and showing that, for every $p\in\mathcal{P}$,
\begin{equation}\label{K}
\mathcal{K}(p^{\rho_0})=\mathcal{K}(p)\rho_0.
\end{equation}
Indeed, given $p\in\mathcal{P}$, it is immediately checked that, for every $x,y\in N$,  we have
$w_{p^{\rho_0}}(x,y)=w_{p}(y,x)$.
As a consequence, for every $q\in \mathcal{L}(N)$, we have that
\[
k_{p^{\rho_0}}(q\rho_0)=\sum_{x\succ_{q\rho_0} y}w_{p^{\rho_0}}(x,y)=\sum_{y\succ_q x}w_p(y,x)=k_p(q).
\]
Given now $\hat{q}\in\mathcal{K}(p)\rho_0$, note that $\hat{q}=q^*\rho_0$ for a suitable $q^*\in \mathcal{K}(p)$. Then, for every $q\in \mathcal{L}(N)$, we have $k_p(q^*)\geq k_p(q)$ and therefore $k_{p^{\rho_0}}(q^*\rho_0)=k_p(q^*)\geq k_p(q)=k_{p^{\rho_0}}(q\rho_0)$.
Since $\mathcal{L}(N)\rho_0=\mathcal{L}(N)$, that means $\hat{q}=q^*\rho_0\in \mathcal{K}(p^{\rho_0})$. Thus,
$\mathcal{K}(p)\rho_0\subseteq \mathcal{K}(p^{\rho_0})$. The same argument applied to $p^{\rho_0}$ gives $\mathcal{K}(p^{\rho_0})\rho_0\subseteq \mathcal{K}(p).$ It follows that $\mathcal{K}(p^{\rho_0})\subseteq \mathcal{K}(p)\rho_0$, which completes the proof of \eqref{K}.

 Assume now that $Kem(p)=\{x\}$ and $Kem(p^{ \rho_0})=\{y\}$ for some $p\in\mathcal{P}$ and $x,y\in N$. Pick $q^*\in \mathcal{K}(p)$ and note that $\mathrm{rank}_{q^*}(x)=1$. On the other hand, by \eqref{K}, $q^*\rho_0\in \mathcal{K}(p^{\rho_0})$, so that $\mathrm{rank}_{q^*\rho_0}(y)=1$, that is, $\mathrm{rank}_{q^*}(y)=n$. Hence the alternatives $x$ and $y$ are ranked differently by $q^*$, which implies $x\neq y$.
\end{proof}

The following result is proved by Bubboloni and Gori (2016, Theorem A).

\begin{proposition}\label{minimax}
$Min$ is immune to the reversal bias if and only if $h\le 3$ or $n\le 3$ or $(h,n)\in\{(4,4), (5,4),(7,4),(5,5)\}$.
\end{proposition}

\begin{proposition}
$Par$ is not resolute.
\end{proposition}

\begin{proof}
Consider $p\in\mathcal{P}$ such that $\mathrm{rank}_{p_1}(1)=1$ and $\mathrm{rank}_{p_2}(2)=1$. Then $\{1,2\}\subseteq Par(p)$ so that $Par$ is not resolute.
\end{proof}

\begin{proposition}
Let $C\in\{Bor,Cop,Min,Kem\}$. Then $C$ is resolute if and only if $n=2$ and $h$ is odd.
\end{proposition}

\begin{proof}
If $n=2$ and $h$ is odd, then $Bor$, $Cop$, $Min$ and $Kem$ agree with the simple majority so that they are resolute.
We are then left with proving that if $h$ is even or $n\ge 3$, then none among $Bor$, $Cop$ , $Min$, $Kem$  is resolute.

Assume at first $h$ even. Define
\[
q_1=[1,2,(3),\ldots,(n)]^T,\quad q_2=[2,1,(3),\ldots,(n)]^T
\]
and consider any preference profile $p\in\mathcal{P}$ such that
\[
\begin{array}{l}
|\{i\in H: p_i=q_1\}|=\frac{h}{2},\quad
 |\{i\in H: p_i=q_2\}|=\frac{h}{2}.
\end{array}
\]
It is immediate to verify that $Bor(p)=Cop(p)=Min(p)=\{1,2\}$. Moreover, it can be checked that
\[
k_p(q_1)=k_p(q_2)=\max_{q\in\mathcal{L}(N)}k_p(q)= (n^2-n-1)\frac{h}{2},
\]
so that, since $Kem$ is efficient, $Kem(p)=\{1,2\}$.

Assume now $n\ge 3$ and $h$ odd. Then there exist $r,t\in\{0,1\}$ with $t\le r$ and $k\in \mathbb{N}_0$ such that $h=3+2r+2t+6k$. Define
\[
\begin{array}{l}
q_1=[1,2,3,(4),\ldots,(n)]^T,\quad q_2=[1,3,2,(4),\ldots,(n)]^T,\quad
q_3=[2,1,3,(4),\ldots,(n)]^T,\\
\vspace{-3mm}\\
q_4=[2,3,1,(4),\ldots,(n)]^T,\quad
q_5=[3,1,2,(4),\ldots,(n)]^T,\quad q_6=[3,2,1,(4),\ldots,(n)]^T,
\end{array}
\]
and consider any preference profile $p\in\mathcal{P}$ such that
\[
\begin{array}{l}
|\{i\in H: p_i=q_1\}|=1+k+r,\quad
 |\{i\in H: p_i=q_2\}|=k+t,\quad
|\{i\in H: p_i=q_3\}|=k,\\
\vspace{-3mm}\\
 |\{i\in H: p_i=q_4\}|=1+k+t,\quad
|\{i\in H: p_i=q_5\}|=1+k,\quad
|\{i\in H: p_i=q_6\}|=k+r.
\end{array}
\]
It is immediate to verify that  $Bor(p)=Cop(p)=Min(p)=\{1,2,3\}$. Moreover, it can be checked that
\[
k_p(q_1)=k_p(q_4)=k_p(q_5)=\max_{q\in\mathcal{L}(N)}k_p(q)=5+3r+3t+9k+(n^2-n-6)\frac{h}{2},
\]
so that, since $Kem$ is efficient, $Kem(p)=\{1,2,3\}$.
\end{proof}

\section{Main problem and results}\label{main}

As shown in the previous section resoluteness is not generally satisfied by classical social choice correspondences. If a {\sc scc} $C$ is not resolute, then it admits different resolute refinements, each of them naturally interpreted as a tie-breaking rule for $C$. As a consequence, one may wonder whether it is possible to find resolute refinements having special and desirable properties. An interesting result about the Pareto {\sc scc} and the properties of anonymity and neutrality is proved by Moulin (1983, p.23).
\begin{theorem}\label{moulin-teo}
$Par$ admits an anonymous and neutral resolute refinement if and only if
\begin{equation}\label{moulin}
gcd(h,n!)=1
\end{equation}
\end{theorem}
An important consequence of Theorem \ref{moulin-teo} is that any efficient {\sc scc} has anonymous and neutral resolute refinements only if \eqref{moulin}. Unfortunately, since  \eqref{moulin} is a very strong arithmetical condition on the number of individuals and the number of alternatives, in  most cases, no anonymous and neutral refinement is available\footnote{Mainly under the assumption  \eqref{moulin}, an analysis of anonymous, neutral and monotonic resolute refinements of {\sc scc}s has been recently carried on by Do\u gan and Giritligil (2015). It is also worth mentioning the contribution of Campbell and Kelly (2015) who show that if $n>h$ (so that \eqref{moulin} fails) and anonymous, neutral and resolute {\sc scc}s exist, then those {\sc scc}s exhibit even more undesirable behaviours than inefficiency.}. When a {\sc scc} does not have anonymous and neutral resolute refinements, one may however focus on those resolute refinements satisfying suitable weaker versions of anonymity and neutrality. Indeed, the present paper is devoted to that type of inquiry. More precisely, having in mind the properties discussed in Section \ref{scc}, we address the following problem.

\vspace{2mm}
\noindent {\bf Main problem.} {\it
Given a  {\sc scc} $C$, a partition $Y$ of individuals  and a partition $Z$ of alternatives, find conditions on $C$, $Y$ and $Z$ assuring that $C$ admits a resolute refinement which is $Y$-anonymous and $Z$-neutral [and immune to the reversal bias]. }
\vspace{2mm}

The first result we propose provides arithmetical conditions on the structure of the partitions that are necessary for the existence of resolute refinements of the Pareto {\sc scc} which are anonymous and neutral with respect to those partitions [and immune to the reversal bias].

\begin{theorem}\label{super}Let $Y=\{Y_j\}_{j=1}^s$ be a partition of $H$, and $Z=\{Z_k\}_{k=1}^t$ be a partition of $N$ with $|Z_{k^*}|=\max\{ |Z_k| \}_{k=1}^t$.
\begin{itemize}
\item[(i)] If $Par$ admits a resolute refinement which is $Y$-anonymous and $Z$-neutral, then
	\begin{equation}\label{f1}
	\gcd\left(\gcd(|Y_j|)_{j=1}^s, \,|Z_{k^*}|!\right)=1.
	\end{equation}
\item[(ii)]  If $Par$ admits a resolute refinement which is $Y$-anonymous, $Z$-neutral and immune to the reversal bias, then
	\begin{equation}\label{f2}
	\gcd\left(\gcd(|Y_j|)_{j=1}^s, \,\mathrm{lcm}(|Z_{k^*}|!,2)\right)=1.
	\end{equation}
	\end{itemize}
\end{theorem}

\begin{proof} (i) Let $C$ be resolute, efficient, $Y$-anonymous and $Z$-neutral. Assume by contradiction that \eqref{f1} does not hold true. Then there exists a prime number $\pi$ which divides $\gcd\left(\gcd(|Y_j|)_{j=1}^s, \,|Z_{k^*}|!\right)$.
Then, for every $j\in\{1,\ldots,s\}$, $\pi \mid |Y_j|$ and $\pi \le |Z_{k^*}|$. Consider $\pi$ distinct alternatives  $x_1,\ldots,x_\pi\in Z_{k^*}$ and denote by $y_1,\ldots, y_{n-\pi}$ the remaining alternatives. For every $j\in\{1,\ldots,s\}$, let  $h_j=|Y_j|$ and $i^j_1, \ldots, i^j_{h_j}$ be the list of all the elements in $Y_j$. Define
\begin{equation}\label{phi-f}
\varphi=(i^1_1 \ldots i^1_{h_1})(i^2_1 \ldots i^2_{h_2})\ldots(i^s_1 \ldots i^s_{h_s})\in S_h
\end{equation}
and $\psi=(x_1\dots x_\pi)\in S_n$. Note that  $\varphi\in V(Y)$ and  $\psi\in W(Z)$.
Consider then the preference relation $p_0\in \mathcal{L}(N)$ defined by
\[
x_1\succ_{p_0}\ldots \succ_{p_0} x_{\pi}\succ_{p_0} y_1\succ_{p_0}\ldots\succ_{p_0} y_{n-\pi},
\]
and the preference profile $p$ defined, for every $j\in\{1,\ldots,s\}$ and $r\in\{1,\ldots,h_j\}$, by $p_{i^j_r}=\psi^{r-1}p_0$.
A simple check shows that $Par(p)= \{x_1,\ldots,x_\pi\}$  and that
\[
p=(p^{(\varphi,id,id)})^{(id,\psi,id)}.
\]
Let now $x^*\in N$ be such that $C(p)=\{x^*\}$. Then
\[
\{x^*\}=C(p)=C\left((p^{(\varphi,id,id)})^{(id,\psi,id)}\right)=\psi C(p^{(\varphi,id,id)})=\psi C(p)=\{\psi(x^*)\}\subseteq Par(p).
\]
Thus $x^*$ is a fixed point of $\psi$ belonging to $\{x_1,\ldots,x_\pi\}$, a contradiction.

 (ii) Let $C$ be resolute, efficient, $Y$-anonymous, $Z$-neutral and immune to the reversal bias.  Assume, by contradiction, that \eqref{f2} does not hold true.
Then there exists a prime number $\pi$ which divides $\gcd\left(\gcd(|Y_j|)_{j=1}^s, \,\mathrm{lcm}(|Z_{k^*}|!,2)\right)$.
If $\pi \ge 3$, then we have that, for every $j\in\{1,\ldots,s\}$, $\pi \mid |Y_j|$ and $\pi \le |Z_{k^*}|$. Thus, we proceed exactly as in the proof of (i) and find the contradiction.
If instead $\pi=2$, then, for every $j\in\{1,\ldots,s\}$,  $h_j=|Y_j|$ is even. Let $i^j_1, \ldots, i^j_{h_j}$ be the list of all the elements in $Y_j$. Let $\varphi$ defined as in \eqref{phi-f} and $\psi=id\in S_n$. Note that  $\varphi\in V(Y)$ and $\psi\in W(Z)$.
Consider
the preference profile $p$ defined, for every $j\in\{1,\ldots,s\}$ and $r\in\{1,\ldots,h_j\}$, by $p_{i^j_r}=\rho_0^{r-1}$.
A simple check shows
that
\[
p=(p^{(id,id,\rho_0)})^{(\varphi,id,id)}.
\]
Then
\[
C(p)=C\left((p^{(id,id,\rho_0)})^{(\varphi,id,id)}\right)=C(p^{(id,id,\rho_0)})\neq C(p),
\]
a contradiction.
\end{proof}

Note that \eqref{f2} obviously implies \eqref{f1}, and that  \eqref{f1} and \eqref{f2} are equivalent if one among the elements of $Z$ has size greater than $1$. Moreover, considering $Y=\{H\}$ and $Z=\{N\}$, both \eqref{f1} and \eqref{f2} are equivalent to \eqref{moulin}.
As a consequence, Theorem \ref{super}(i) implies the ``only if'' part of Theorem \ref{moulin-teo}.

Let us present now the main result of the paper. It shows that, considering partitions of individuals and alternatives satisfying the arithmetical conditions described in Theorem \ref{super}, there exists a resolute refinement which is anonymous and neutral with respect to the given partitions [and immune to the reversal bias] for any {\sc scc} having the same properties. Differently from Theorem \ref{super}, the proof of this result is quite technical and will be  presented in Section \ref{appl} as a consequence of the theory developed in Section \ref{gen-teo}.

\begin{theorem}\label{super-super}
Let $Y=\{Y_j\}_{j=1}^s$ be a partition of $H$, and $Z=\{Z_k\}_{k=1}^t$ be a partition of $N$ with $|Z_{k^*}|=\max\{ |Z_k| \}_{k=1}^t$.
\begin{itemize}
	\item[(i)] If $C\in\mathfrak{C}$ is $Y$-anonymous and $Z$-neutral and \eqref{f1} holds true, then $C$ admits a resolute refinement which is  $Y$-anonymous, $Z$-neutral.
	\item[(ii)]If $C\in\mathfrak{C}$ is  $Y$-anonymous, $Z$-neutral and immune to the reversal bias and \eqref{f2} holds true, then $C$ admits a resolute refinement which is  $Y$-anonymous, $Z$-neutral and immune to the reversal bias.
\end{itemize}
\end{theorem}
Note that if \eqref{moulin} holds true, then \eqref{f1} and \eqref{f2} are both satisfied for $Y=\{H\}$ and $Z=\{N\}$. As a consequence, from Theorem \ref{super-super} we get that \eqref{moulin} implies that every {\sc scc} which is anonymous and neutral [and immune to the reversal bias] admits a resolute refinement which is anonymous and neutral [and immune to the reversal bias]. That shows, in particular, that Theorem \ref{super-super}(i) implies the ``if'' part of Theorem \ref{moulin-teo}.

In order to better understand how Theorem \ref{super-super} can be applied in concrete situations, let us consider a committee whose purpose is to select a unique winner among a certain set of candidates. Assume that the members of a committee have already found an agreement on the use of a certain {\sc scc} $C$ with $C$ anonymous and neutral [and immune to the reversal bias]. If $C$ is not resolute, then the committee members need to determine a tie-breaking rule, that is, a resolute refinement of $C$. By Theorem \ref{super-super}, if the characteristics of committee members and candidates naturally suggest to divide them in groups satisfying \eqref{f1} [\eqref{f2}], a resolute refinement of $C$ which is anonymous and neutral with respect to the considered partitions [and immune to the reversal bias] can be designed.
As a remarkable case, assume that the committee has an individual, say individual $i$, having a special role in the committee. That happens, for instance, when the committee has a president. In that case, it is reasonable to assume that the decision power of the president may be potentially  different from the one of any other member in the committee. That leads to consider the partition of $H$ given by $Y=\{\{i\},H\setminus\{i\}\}$ which determines the groups of people equally influencing the outcome of the decision process. Considering now the partition $Z=\{N\}$ of $N$, that is giving no exogenous advantage to any alternative, we have that $Y$ and $Z$ satisfy \eqref{f2} and so also \eqref{f1}. Then Theorem \ref{super-super} applies and we know that there exists a resolute refinement of $C$ which is $Y$-anonymous and neutral [and immune to the reversal bias].

As a final remark, we stress that, in general, resolute refinements identified by Theorem \ref{super-super} may be a lot, so that it can be difficult to discriminate among them. However, the theory developed in Section \ref{gen-teo} provides a method to potentially build and count all those refinements. That makes the comparison among them much easier. In order to describe how such a method works, in Section \ref{example53}, we build and count all resolute refinements which are anonymous, neutral and immune to the reversal bias of the Pareto, the Borda, the Copeland, the Minimax and the Kemeny {\sc scc}s when individuals are five and alternatives are three.
In Section \ref{example33} we consider instead the case with three individuals and three alternatives and we analyse, for each of the previously mentioned {\sc scc}s, all the resolute refinement which are  $\{\{1,2\}, \{3\}\}$-anonymous, neutral and immune to the reversal bias. In that example, individual 3 is distinguished from individuals 1 and 2 who instead are indistinguishable.

\section{General theory}\label{gen-teo}

In this large section we develop the general theory behind our results. The concept of action of a group on a set is the main tool used (see
 Jacobson, 1974, Section 1.12).


\subsection{$U$-consistent resolute {\sc scc}s }

 Let $\mathfrak{F}$ denote the set of resolute {\sc scc}s.  Clearly $\mathfrak{F}\subseteq \mathfrak{C}$ and each resolute {\sc scc}s can be naturally identified with a social choice function, that is, a function $f$ from $\mathcal{P}$ to $N$.
We will adopt that identification throughout the rest of the paper.

Let $C\in \mathfrak{C}$. Denote by $\mathfrak{C}_C$ the set of refinements of $C,$  and by $\mathfrak{F}_C$ the set $\mathfrak{F}\cap \mathfrak{C}_C$  of the resolute refinements of $C$. Each resolute refinement of $C$ is identified with a social choice function $f:\mathcal{P}\rightarrow N$ such that, for every $p\in \mathcal{P}$,
\begin{equation}\label{functionC}
\begin{array}{l}
f(p)\in C(p).
\end{array}
\end{equation}

Let $U$ be a subgroup  of $G$.
We say that $C$ is $U$-{\it consistent} if, for every $p\in \mathcal{P}$ and $(\varphi,\psi,\rho)\in U$,
\begin{eqnarray}
&&C(p^{(\varphi,\psi,\rho)})=\psi  C(p)\quad\mbox{ if } \rho=id, \label{sccU1}\\
&&C(p^{(\varphi,\psi,\rho)})\neq\psi  C(p) \quad\mbox{ if } \rho=\rho_0 \mbox{ and } |C(p)|=1.\label{sccU2}
\end{eqnarray}
Note that condition \eqref{sccU2} is equivalent to
\begin{eqnarray}
&&C(p^{(\varphi,\psi,\rho)})\neq\psi  C(p) \quad\mbox{ if } \rho=\rho_0 \mbox{ and } |C(p)|=|C(p^{(\varphi,\psi,\rho_0)})|=1.\label{sccU2bis}
\end{eqnarray}
The set of $U$-consistent {\sc scc}s is denoted by $\mathfrak{C}^U$; the set $\mathfrak{F}\cap \mathfrak{C}^U$ of $U$-consistent resolute {\sc scc}s is denoted by $\mathfrak{F}^U$. Each  {\sc scc}  in $ \mathfrak{F}^U$  is identified with a social choice function $f:\mathcal{P}\rightarrow N$ such that, for every $p\in \mathcal{P}$ and $(\varphi,\psi,\rho)\in U$,

\begin{equation}\label{functionU}
\begin{array}{l}
f(p^{(\varphi,\psi,id)})=\psi  f(p) \quad\mbox{ if } \rho=id,\\
\vspace{-3mm}\\
 f(p^{(\varphi,\psi,\rho_0)})\neq\psi f(p) \quad\mbox{ if } \rho=\rho_0.
\end{array}
\end{equation}

We also set $\mathfrak{C}_C^U=\mathfrak{C}_C\cap \mathfrak{C}^U$ and
$\mathfrak{F}_C^U=\mathfrak{F}_C\cap \mathfrak{F}^U$.
Each  {\sc scc}  in $ \mathfrak{F}_C^U$ will be called a $U$-consistent resolute refinement of $C$ and
identified with a social choice function $f:\mathcal{P}\rightarrow N$ such that, for every $p\in \mathcal{P}$, both \eqref{functionC} and \eqref{functionU} hold true.

The concept of consistency of a {\sc scc} with respect to a subgroup $U$ of $G$  is able to catch
interesting requirements for {\sc scc}s through a suitable choice of the subgroup $U$.
In particular, for every partition $Y$ of $H$, $C$ is $Y$-anonymous if and only if $C\in\mathfrak{C}^{V(Y)\times \{id\}\times \{id\}}$; for every partition $Z$ of $N$, $C$ is $Z$-neutral if and only if $C\in\mathfrak{C}^{\{id\}\times W(Z)\times \{id\}}$; $C$ is immune to the reversal bias if and only if $C\in\mathfrak{C}^{\{id\}\times \{id\}\times \Omega}$.

\subsection{The action of $G$ on the set of preference profiles}\label{def-not}

The following proposition, proved in Bubboloni and Gori (2015, Proposition 2), shows that any subgroup $U$ of $G$ naturally acts on the set of preference profiles $\mathcal{P}$. Recall that this means that there exists a homomorphism from $U$ to $\mathrm{Sym}$.

\begin{proposition}\label{action-l}
Let $U\le G$. Then:
\begin{itemize}
	\item[(i)] for every $p\in\mathcal{P}$ and $(\varphi_1,\psi_1,\rho_1),(\varphi_2,\psi_2,\rho_2)\in U$, we have
\begin{equation}\label{action-e}
p^{\,(\varphi_1\varphi_2,\psi_1\psi_2,\rho_1\rho_2)}= (p^{\,(\varphi_2,\psi_2,\rho_2)})^{(\varphi_1,\psi_1,\rho_1)};
\end{equation}
	\item[(ii)] the function $\alpha:U\to \mathrm{Sym}(\mathcal{P})$ defined, for every $(\varphi,\psi,\rho)\in U$, by
\[
\alpha(\varphi,\psi,\rho):\mathcal{P}\to\mathcal{P},\quad p\mapsto p^{(\varphi,\psi,\rho)},
\]
is well defined and it is an action of the group $U$ on the set $\mathcal{P}$.
\end{itemize}
\end{proposition}

Proposition \ref{U,V-corr} below is a first interesting consequence of Proposition \ref{action-l}. In particular, it says that, given a
{\sc scc} $C$, a partition $Y$ of $H$ and a partition $Z$ of $N$, we have that
 $C$ is $Y$-anonymous and $Z$-neutral if and only if $C\in\mathfrak{C}^{V(Y)\times W(Z)\times \{id\}}$; $C$ is $Y$-anonymous and immune to the reversal bias if and only if $C\in\mathfrak{C}^{V(Y)\times \{id\}\times \Omega}$; $C$ is $Z$-neutral and immune to the reversal bias if and only if $C\in\mathfrak{C}^{\{id\}\times W(Z)\times \Omega}$; $C$ is $Y$-anonymous, $Z$-neutral and immune to the reversal bias if and only if $C\in\mathfrak{C}^{G}$.

Before stating Proposition \ref{U,V-corr}, recall that if $X$ is a subset of a group $G$, the subgroup of $G$ generated by $X$ is defined as the intersection of all the subgroups of $G$ containing $X$ and it is denoted by $\langle X\rangle.$ It is well known that $\langle X\rangle$ consists of all the finite products of elements in $X$. If $X_1,X_2$ are subsets of $G$, we write $\langle X_1,X_2\rangle$ instead of $\langle X_1\cup X_2\rangle$.
 For further details see Jacobson (1974, Section 1.5).

\begin{proposition}\label{U,V-corr}
For every $i\in\{1,2\}$, let $Z_i\le S_h\times S_n$, $R_i\le \Omega$ and $U_i=Z_i\times R_i$. Then $\mathfrak{C}^{U_1}\cap \mathfrak{C}^{U_2}=\mathfrak{C}^{\langle U_1,U_2\rangle}$. In particular, $\mathfrak{F}^{U_1}\cap\, \mathfrak{F}^{U_2}=\mathfrak{F}^{\langle U_1,U_2\rangle}$.
\end{proposition}

\begin{proof}  Since $\langle U_1,U_2\rangle\leq G$ contains both $U_1$ and $U_2$, we  immediately get $\mathfrak{C}^{\langle U_1,U_2\rangle}\subseteq \mathfrak{C}^{U_1}\cap \mathfrak{C}^{U_2}.$ Let us now fix $C\in\mathfrak{C}^{U_1}\cap \mathfrak{C}^{U_2}$ and prove
 $C\in \mathfrak{C}^{\langle U_1,U_2\rangle}$.
Define, for every  $k\in \mathbb{N}$, the set $\langle U_1,U_2\rangle_k$ of the elements in $\langle U_1,U_2\rangle$ that can be written as product of $k$ elements of $U_1\cup U_2$. Then 
 we have $\langle U_1,U_2\rangle=\bigcup_{k\in\mathbb{N}}\langle U_1,U_2\rangle_k$ and to get $C\in \mathfrak{C}^{\langle U_1,U_2\rangle}$ it is enough to show the two following facts:
\begin{itemize}
\item[(a)] for every $k\in\mathbb{N},$
\begin{equation}\label{A1}
\mbox{for every } p\in\mathcal{P} \mbox{ and } g=(\varphi, \psi, id)\in \langle U_1,U_2\rangle_k,\mbox{ \eqref{sccU1} holds true;}
\end{equation}
\item[(b)]  for every $k\in\mathbb{N},$ $ p\in\mathcal{P}$ and $g=(\varphi, \psi, \rho_0)\in \langle U_1,U_2\rangle_k$, \eqref{sccU2} holds true.
\end{itemize}
First of all, for every $g=(\varphi,\psi,\rho)\in G$, define $\overline{g}=(\varphi,\psi,id)\in G$. We start
showing that, for every $k\in \mathbb{N}$,
\begin{equation}\label{B}
 \  g\in \langle U_1,U_2\rangle_k\ \mbox{ implies} \ \overline{g}\in \langle U_1,U_2\rangle_k.
 \end{equation}
 If $\rho=id$, there is nothing to prove. So assume $\rho=\rho_0$.
Since, for every  $i\in\{1,2\}$, we have that $U_i=Z_i\times  R_i$ with $Z_i\le S_h\times S_h$ and $R_i\le  \Omega$, then \eqref{B} surely holds for $k=1$. If $k\geq 2$, pick $g=g_1\cdots g_k=(\varphi,\psi,\rho_0)\in \langle U_1,U_2\rangle_k$, where $g_1,\ldots,g_k\in U_1\cup U_2$. Since $\rho_0$ has order two, the number of $j\in\{1,\dots,k\}$ such that the third component of $g_j$ is $\rho_0$ is odd. Pick $j\in\{1,\dots,k\}$ such that $g_j=(\varphi_j,\psi_j,\rho_0)$.
By the case $k=1$, we have that  $\overline{g}_j=(\varphi_j,\psi_j,id)\in U_1\cup U_2$, so that
$\overline{g}=g_1\ldots g_{j-1}\overline{g}_jg_{j+1}\dots g_k\in \langle U_1,U_2\rangle_k$ and its first and second components are equal to those of $g.$ Moreover, the number of factors  in $\overline{g}$ having as third component $\rho_0$ is even, which gives $\overline{g}=(\varphi,\psi,id).$

We now show (a), by induction on $k.$ If $k=1$, we have $g\in\langle U_1,U_2\rangle_1= U_1\cup U_2$ and so \eqref{A1} is guaranteed by $C\in\mathfrak{C}^{U_1}\cap \mathfrak{C}^{U_2}$.
Assume \eqref{A1}  up to some $k\in\mathbb{N}$ and show that it holds also for $k+1$. Let $p\in\mathcal{P}$ and $g=(\varphi,\psi,id)\in \langle U_1,U_2\rangle_{k+1}$. Then there exist $g_*=(\varphi_*,\psi_*,\rho_*)\in \langle U_1,U_2\rangle_{k}$ and $g_1=(\varphi_1,\psi_1,\rho_1)\in U_1\cup U_2$
such that $g=g_1g_*=(\varphi_1\varphi_*,\psi_1\psi_*,\rho_1\rho_*)$.
We want to show that $C(p^{g})=\psi_1\psi_*C(p)$.
Note that $g=\overline{g}_1\overline{g}_*$ and that, by \eqref{B},  $\overline{g}_*\in \langle U_1,U_2\rangle_{k}$
and $\overline{g}_1\in U_1\cup U_2.$
Then, using \eqref{action-e} and applying the inductive hypothesis for \eqref{A1} both to $\overline{g}_1$ and to $\overline{g}_*$,  we get
$C(p^{g})=C(p^{\overline{g}_1\overline{g}_*})=C((p^{\overline{g}_*})^{\overline{g}_1})=\psi_1C(p^{\overline{g}_*})=\psi_1\psi_*C(p)$.

We next show (b).
Let $k\in\mathbb{N}$,
$p\in\mathcal{P}$, $g=(\varphi, \psi, \rho_0)\in \langle U_1,U_2\rangle_{k}$ and $|C(p)|=1$. We need to show that $C(p^g)\neq \psi C(p).$
First of all note that, since $\langle U_1,U_2\rangle$ contains an element with third component $\rho_0$, then we necessarily have $R_1=\Omega$ or  $R_2=\Omega$, so that $(id,id,\rho_0)\in U_1\cup U_2$.
Moreover, we can express $g$ as $g=\overline{g}\,(id,id,\rho_0)$ and,  by \eqref{B}, $\overline{g} \in \langle U_1,U_2\rangle_{k}$.
Thus, by \eqref{action-e}  and (a), we have $C(p^g)=C((p^{(id,id,\rho_0)})^{\overline{g}})=\psi C(p^{(id,id,\rho_0)}).$
On the other hand, since $(id,id,\rho_0)\in U_1\cup U_2$ and $C\in\mathfrak{C}^{U_1}\cap \mathfrak{C}^{U_2}$, we get
 $C(p^{(id,id,\rho_0)})\neq C(p)$ and so $C(p^g)=\psi C(p^{(id,id,\rho_0)})\neq \psi C(p)$ as required.

Finally, as a consequence of $\mathfrak{C}^{U_1}\cap\, \mathfrak{C}^{U_2}=\mathfrak{C}^{\langle U_1,U_2\rangle}$, we also have that  $\mathfrak{F}^{U_1}\cap \mathfrak{F}^{U_2}=\mathfrak{C}^{U_1}\cap \mathfrak{C}^{U_2}\cap \mathfrak{F}=\mathfrak{C}^{\langle U_1,U_2\rangle}\cap \mathfrak{F}=  \mathfrak{F}^{\langle U_1,U_2\rangle}$.
\end{proof}

As an immediate consequence of Propositions \ref{ean}, \ref{others}, \ref{minimax} and \ref{U,V-corr} we get the following result.
\begin{proposition}\label{moulin2}	
\begin{itemize}
\item[(i)] For every $C\in\{Par,Bor,Cop,Min,Kem\}$, $C\in\mathfrak{C}_{Par}^{S_h\times S_n\times \{id\}}$.
	\item[(ii)] For every $C\in\{Par,Bor,Cop,Kem\}$, $C\in\mathfrak{C}_{Par}^{G}$.
	\item[(iii)] $Min\in \mathfrak{C}_{Par}^G$ if and only if one of the following conditions holds true:
	\begin{itemize}
		\item[(a)] $h\le 3$;
		\item[(b)] $n\le 3$;
		\item[(c)] $(h,n)\in\{(4,4),(5,4),(7,4),(5,5)\}$.
	\end{itemize}
\end{itemize}
\end{proposition}

 Proposition \ref{action-l} also allows to use notation and results concerning the action of a group on a set. We recall the basic facts that we are going to use.

Fix $U\le G$.
For every $p\in\mathcal{P}$, the set $p^U=\{p^g\in \mathcal{P}: g\in U\}$ is called the $U$-orbit of $p$ and the subgroup of $U$ defined by
$\mathrm{Stab}_U(p)=\{g\in U : p^g=p \}$ is called the stabilizer of $p$ in $U$. It is well known that
the set  $\mathcal{P}^U=\{p^U:p\in\mathcal{P}\}$ of the $U$-orbits is a partition of $\mathcal{P}$.
We use $\mathcal{P}^U$ as set of indexes and denote its elements with $j$. A vector $(p^j)_{j\in\mathcal{P}^U}\in\times_{j\in \mathcal{P}^U}\mathcal{P}$
is called a system of representatives of the $U$-orbits if, for every $j\in\mathcal{P}^U$, $p^j\in j$.
The set of the systems of representatives of the $U$-orbits is denoted by $\mathfrak{S}(U)$.
If $(p^j)_{j\in \mathcal{P}^U}\in \mathfrak{S}(U)$, then, for every $p\in \mathcal{P}$, there exist $ j\in\mathcal{P}^U$ and $(\varphi,\psi,\rho)\in U$ such that $p=p^{j\,(\varphi,\psi,\rho)}$. Note that if $p^{j_1\,(\varphi_1,\psi_1,\rho_1)}=p^{j_2\,(\varphi_2,\psi_2,\rho_2)}$ for some $j_1,j_2\in \mathcal{P}^U$ and some $(\varphi_1,\psi_1,\rho_1), (\varphi_2,\psi_2,\rho_2)\in U$, then $j_1= j_2$ and, by \eqref{action-e},
$(\varphi_2^{-1}\varphi_1,\psi_2^{-1}\psi_1,\rho_2^{-1}\rho_1)\in \mathrm{Stab}_U(p^{j_1})$.

The stabilizer of $p$ in $U$ evolves in a natural way through the action. Namely
 for every  $p\in\mathcal{P}$ and $g\in U$, we have
\[
\mathrm{Stab}_U(p^{g})=g\,\mathrm{Stab}_U(p)g^{-1}.
\]
This implies that if $V$ is a normal subgroup of $U$ and $p\in\mathcal{P}$, then  $\mathrm{Stab}_U(p)\leq V$ if and only if $\mathrm{Stab}_U(p^{g})\leq V$  for all  $g\in U$. Now, being $S_h\times S_n\times \{id\}$ normal in $G$, by an elementary group theory result, we have that $U\cap (S_h\times S_n\times \{id\})$ is normal in $U$. Thus, the above argument guarantees that, for every $j\in \mathcal{P}^U$, exactly one of the two following conditions holds true:
\begin{itemize}
	\item[-] for every $p\in j$, $\mathrm{Stab}_U(p)\le  S_h\times S_n\times \{id\}$;
	\item[-] for every $p\in j$, $\mathrm{Stab}_U(p)\not\le S_h\times S_n\times \{id\}$.
\end{itemize}
We then define
\begin{eqnarray}
&&\mathcal{P}_1^U=\left\{j\in\mathcal{P}^U: \forall p\in j,\, \mathrm{Stab}_U(p)\le S_h\times S_n\times \{id\}\right\}\label{P1},\\
&&\mathcal{P}_2^U=\left\{j\in\mathcal{P}^U: \forall p\in j, \,\mathrm{Stab}_U(p)\not \le S_h\times S_n\times \{id\}\right\}\label{P2}.
\end{eqnarray}
Of course, $\mathcal{P}_1^U\cup\mathcal{P}_2^U=\mathcal{P}^U$ and $\mathcal{P}_1^U\cap\mathcal{P}_2^U=\varnothing$. In particular, $\mathcal{P}_1^U$ and $ \mathcal{P}_2^U$ cannot be both empty. Obviously, if $U\leq S_h\times S_n\times \{id\}$, then $\mathcal{P}_2^U=\varnothing$ and $\mathcal{P}^U=\mathcal{P}_1^U\neq \varnothing.$

The sets $\mathcal{P}_1^U$ and $\mathcal{P}_2^U$ play an important role to check whether two given $U$-consistent {\sc scc}s are equal, as shown by the following results.

\begin{proposition}\label{rappresentanti} Let $U\leq S_h\times S_n \times \{id\}$ and $C,C'\in \mathfrak{C}^{U}$. Assume that there exists $(p^j)_{j\in\mathcal{P}^U }\in \mathfrak{S}(U)$
 such that $C(p^j)=C'(p^j)$ for all $j\in \mathcal{P}^U$. Then $C=C'$.
\end{proposition}

\begin{proof} Let $p\in \mathcal{P}$ and show that $C(p)=C'(p)$. We know there exist $j\in\mathcal{P}^U$ and $(\varphi,\psi,id)\in U$ such that $p=p^{j\,(\varphi,\psi,id)}$.
Then,
\begin{equation}\label{chain}
C(p)=C(p^{j\,(\varphi,\psi,id)})=\psi C(p^{j})=\psi C'(p^{j})=C'(p^{j\,(\varphi,\psi,id)})=C'(p).
\end{equation}
\end{proof}

\begin{proposition}\label{rappresentanti2} Let $U\leq G$ such that $U\not\le S_h\times S_n\times \{id\}$ and $C,C'\in \mathfrak{C}^{U}$. Assume that there exist $(p^j)_{j\in\mathcal{P}^U }\in \mathfrak{S}(U)$ and $(\varphi_*,\psi_*,\rho_0)\in U$
 such that $C(p^j)=C'(p^j)$ for all $j\in \mathcal{P}^U$ and $ C(p^{j\,(\varphi_*,\psi_*,\rho_0)})=C'(p^{j\,(\varphi_*,\psi_*,\rho_0)})$ for all $j\in \mathcal{P}_1^U$. Then $C=C'.$
\end{proposition}

\begin{proof} Let $p\in \mathcal{P}$ and show that $C(p)=C'(p)$. Let $j \in \mathcal{P}^U$ be the unique orbit such that $p\in j .$
If there exists $(\varphi,\psi,id)\in U$ such that $p=p^{j\,(\varphi,\psi,id)}$, then we get $C(p)=C'(p)$ operating as in \eqref{chain}.
So, assume that,
\begin{equation}\label{cond}
\mbox{for every } (\varphi,\psi,\rho)\in U  \mbox{ such that } p=p^{j\,(\varphi,\psi,\rho)},  \mbox{ we have } \rho=\rho_0.
\end{equation}
We show that \eqref{cond} implies $\mathrm{Stab}_U(p^j)\le  S_h\times S_n\times \{id\}$.  Indeed, suppose by contradiction that there exists $(\varphi_1,\psi_1,\rho_0)\in \mathrm{Stab}_U(p^j)$. Pick  $(\varphi,\psi,\rho_0)\in U$ such that $p=p^{j\,(\varphi,\psi,\rho_0)}$ and note that, by \eqref{action-e},
\[p=p^{j\,(\varphi,\psi,\rho_0)}=(p^{j\,(\varphi_1,\psi_1,\rho_0)})^{(\varphi,\psi,\rho_0)}=p^{j\,(\varphi\varphi_1,\psi\psi_1,id)}
\]
which contradicts \eqref{cond}. As a consequence, $j\in \mathcal{P}_1^U$ and thus $C(p^{j\,(\varphi_*,\psi_*,\rho_0)})=C'(p^{j\,(\varphi_*,\psi_*,\rho_0)})$. Pick  again $(\varphi,\psi,\rho_0)\in U$ such that $p=p^{j\,(\varphi,\psi,\rho_0)}$ and note that, by \eqref{action-e},
\[p=p^{j\,(\varphi,\psi,\rho_0)}=(p^{j\,(\varphi_*,\psi_*,\rho_0)})^{(\varphi\varphi_*^{-1},\psi\psi_*^{-1},id)}
\]
so that, since $C$ and $C'$ are $U$-consistent, we finally obtain
\[
C(p)=C\left((p^{j\,(\varphi_*,\psi_*,\rho_0)})^{(\varphi\varphi_*^{-1},\psi\psi_*^{-1},id)}\right)=\psi\psi_*^{-1}C(p^{j\,(\varphi_*,\psi_*,\rho_0)})
\]
\[
=\psi\psi_*^{-1}C'(p^{j\,(\varphi_*,\psi_*,\rho_0)})=C'\left((p^{j\,(\varphi_*,\psi_*,\rho_0)})^{(\varphi\varphi_*^{-1},\psi\psi_*^{-1},id)}\right)=C'(p).
\]
\end{proof}

Propositions \ref{rappresentanti} and \ref{rappresentanti2} indicate that the consistency level of a {\sc scc} is a tool for identifying it. Indeed, let $C,C'\in \mathfrak{C}$ and suppose that we are interested to know whether $C=C'$. Once $C$ and $C'$ are proved to be both in $\mathfrak{C}^{U}$ for a suitable $U\leq G$, it suffices to check whether the equality $C(p)=C'(p)$ holds true on a small subset of $\mathcal{P}$, which essentially agrees with a system of $U$-orbits representatives.
Since the largest is $U$, the smallest the number of $U$-orbits is, dealing with the largest possible $U$ reduces the number of checks to be done.

\subsection{Regular subgroups}\label{section-regular}

Bubboloni and Gori (2015) introduce the concept of regular subgroup to deal with symmetric social welfare functions. 
A subgroup $U$  of $G$ is said to be {\it regular} if, for every $p\in\mathcal{P}$,
\begin{equation}\label{property}
\begin{array}{c}
\mbox{there exists }\psi_*\in S_n\mbox{ conjugate  to } \rho_0\mbox{ such that}\\
\vspace{-3mm}\\
\mathrm{Stab}_U(p)\subseteq \left(S_h\times \{id\}\times \{id\}\right)\cup \left( S_h\times \{\psi_*\}\times \{\rho_0\}\right).\\
\end{array}
\end{equation}
Note that, within our notation, two permutations $\sigma_1,\sigma_2\in S_n$ are conjugate if there exists $u\in S_n$ such that $\sigma_1=u \sigma_2 u^{-1}.$
The following result, which is proved in Bubboloni and Gori (2015, Theorem 14 and Lemma 17), identifies an interesting and quite large class of regular subgroups of $G$. 

\begin{theorem}\label{regular}  Let $Y=\{Y_j\}_{j=1}^s$ be a partition of $H$, $Z=\{Z_k\}_{k=1}^t$ be a partition of $N$ with $|Z_{k^*}|=\max\{ |Z_k| \}_{k=1}^t$ and $R\le \Omega$. Then $V(Y)\times W(Z)\times R$ is regular  if and only if
\begin{equation}\label{reg-eq}
\gcd\left(\gcd(|Y_j|)_{j=1}^s, \,\mathrm{lcm}(|Z_{k^*}|!, |R|)\right)=1.
\end{equation}
In particular, $G$ is regular if and only if $\gcd(h,n!)=1$.
\end{theorem}

Theorem \ref{general} below, which is a corollary of Theorems \ref{fu-min-2} and \ref{fu-min-count2} proved in the next sections, clearly shows the importance of concept of regular subgroup in the context of {\sc scc}s.
\begin{theorem}\label{general}
Let $U\le G$ be regular. Then each $U$-consistent {\sc scc} admits a resolute $U$-consistent refinement.
\end{theorem}
Let us now collect some facts about regular subgroups that we are going to use in the sequel. Recall that the subsets $\mathcal{P}_1^U$ and $\mathcal{P}_2^U$ of $\mathcal{P}^U$ are defined in \eqref{P1} and \eqref{P2} respectively.

\begin{lemma}\label{collect} Let $U\leq G$ be regular. 
\begin{itemize}
\item[(i)] \[
\mathcal{P}_1^U=\left\{j\in\mathcal{P}^U: \forall p\in j,\, \mathrm{Stab}_U(p)\le S_h\times\{id\}\times \{id\}\right\},
\]
\[
\mathcal{P}_2^U=\left\{j\in\mathcal{P}^U: \forall p\in j, \,\mathrm{Stab}_U(p)\not \le S_h\times \{id\}\times \{id\}\right\}.
\]
\item[(ii)] If $p\in \mathcal{P}$ is such that $\mathrm{Stab}_U(p)\not \le S_h\times \{id\}\times \{id\},$ then the permutation  $\psi_*\in S_n$ in \eqref{property} is unique.
\item[(iii)]If  $W\leq U$, then $W$ is regular too. In particular, $G$ is regular if and only if each subgroup of $G$ is regular.
\end{itemize}

\end{lemma}

\begin{proof}$(i)$ It is an immediate consequence of the definitions of $\mathcal{P}_1^U$ and $ \mathcal{P}_2^U$  and of regular subgroup.

$(ii)$ Assume that
$\mathrm{Stab}_U(p)\subseteq \left(S_h\times \{id\}\times \{id\}\right)\cup \left( S_h\times \{\psi_*\}\times \{\rho_0\}\right)$ as well as $\mathrm{Stab}_U(p)\subseteq \left(S_h\times \{id\}\times \{id\}\right)\cup \left( S_h\times \{\psi_{**}\}\times \{\rho_0\}\right)$, for suitable $\psi_*,\psi_{**} \in S_n$ and pick $(\varphi,\psi, \rho_0)\in \mathrm{Stab}_U(p).$ Then, we have $\psi=\psi_*$ as well as $\psi=\psi_{**},$ so that $\psi_*=\psi_{**}.$

$(iii)$ Simply observe that, for every $p\in \mathcal{P}$,  $\mathrm{Stab}_W(p)=W\cap\, \mathrm{Stab}_U(p)$.
\end{proof}

In the Appendix, under the assumption that $U$ is a regular subgroup of $G$, we will discuss when $\mathcal{P}_1^U\neq \varnothing$ or $\mathcal{P}_2^U\neq \varnothing$.

\subsection{Existence of $U$-consistent resolute refinements for $U\leq S_h\times S_n\times \{id\}$}

In this section we focus on the set $\mathfrak{F}^U_C$, where $U$ is a regular subgroup of $G$ included in $S_h\times S_n\times \{id\}$, and $C$ is a $U$-consistent {\sc scc}.

\begin{proposition}\label{f-fu}
Let $U\le S_h\times S_n\times \{id\}$ be regular, $(p^j)_{j\in\mathcal{P}^U}\in\mathfrak{S}(U)$ and $C\in\mathfrak{C}^U$. For every $j\in\mathcal{P}^U$, let $x_j\in C(p^j)$.
Then there exists a unique $f\in\mathfrak{F}^U_C$ such that, for every $j\in\mathcal{P}^U$, $f(p^j)=x_j$.
\end{proposition}

\begin{proof}
Let us consider  $f\in\mathfrak{F}$ defined, for every $p\in\mathcal{P}$, as follows. Given $p\in \mathcal{P}$, consider the unique $j\in\mathcal{P}^U$ such that $p\in j$ and the nonempty set $U_p=\{(\varphi,\psi,id)\in U: p=p^{j\,(\varphi,\psi,id)}\}$. Pick $(\varphi,\psi,id)\in U_p$ and let
$f(p)=\psi (x_j)$. We need to prove that the value of $f(p)$ does not depend on the particular element chosen in $U_p$.
Indeed, let $(\varphi_1,\psi_1,id),(\varphi_2,\psi_2,id)\in U_p$ and recall that $(\varphi_2^{-1}\varphi_1,\psi_2^{-1}\psi_1,id)\in \mathrm{Stab}_U(p^j)$.
Since $U$ is regular, that gives  $\psi_1=\psi_2$ and, in particular, $\psi_1 (x_j)=\psi_2(x_j)$.

We show that $f$ satisfies all the desired properties.
First of all, since $U\le G$, we have $(id,id,id)\in U$ and thus the definition of $f$ immediately implies
 $f(p^j)=x_j$.

Let us now prove that $f\in\mathfrak{F}^U$.
Consider then $p\in\mathcal{P}$ and $(\varphi,\psi,id)\in U$ and show that $f(p^{(\varphi,\psi,id)})=\psi f(p)$. Let $p=p^{j\,(\varphi_1,\psi_1,id)}$ for suitable $j\in\mathcal{P}^U$ and $(\varphi_1,\psi_1,id)\in U$. Thus,
$
f(p)=\psi_1(x_j)
$
and, by \eqref{action-e},
$
f(p^{(\varphi,\psi,id)})=f(p^{j\,(\varphi\varphi_1,\psi\psi_1,id)})=\psi\psi_1  (x_j)=\psi f(p).
$

Let us next prove that $f\in\mathfrak{F}_{C}$. Consider then $p\in\mathcal{P}$ and show that $f(p)\in C(p)$. Let $p=p^{j\,(\varphi_1,\psi_1,id)}$ for suitable $j\in\mathcal{P}^U$ and $(\varphi_1,\psi_1,id)\in U$. Thus,
$
f(p)=\psi_1(x_j)
$ and, since
$C$ is $U$-consistent,
$
\psi_1(x_j)\in \psi_1C(p^j)=C(p^{j\,(\varphi_1,\psi_1,id)})=C(p).
$

Finally, in order to prove uniqueness, let $f'\in\mathfrak{F}^U_C\subseteq \mathfrak{C}^U$ such that
 $f'(p^j)=x_j$ for all $j\in\mathcal{P}^U$. Then $f'$ and $f$ coincides on $(p^j)_{j\in\mathcal{P}^U}\in\mathfrak{S}(U)$ and Proposition \ref{rappresentanti} applies giving $f'=f.$
\end{proof}

Let $U\le S_h\times S_n\times \{id\}$ be regular, $(p^j)_{j\in \mathcal{P}^U}\in\mathfrak{S}(U)$ and $C\in\mathfrak{C}^U$. Let
$\Phi: \times_{j\in\mathcal{P}^U} C(p^j)	\to \mathfrak{F}^{U}_C$ be the function which associates with every  $(x_j)_{j\in \mathcal{P}^U}\in \times_{j\in\mathcal{P}^U} C(p^j)$ the unique $f\in\mathfrak{F}^U_C$ defined in Proposition \ref{f-fu}.
Of course, $\Phi$ depends on $U$, $(p^j)_{j\in \mathcal{P}^U}$ and $C$ but we do not emphasize that dependence in the notation. Note also that $\Phi$ is injective.

\begin{theorem}\label{fu-min-2}Let $U\le S_h\times S_n\times \{id\}$ be regular, $(p^j)_{j\in\mathcal{P}^U}\in\mathfrak{S}(U)$  and  $C\in\mathfrak{C}^U$. Then
\[
\mathfrak{F}_{C}^{U}=\Phi\left(\times_{j\in\mathcal{P}^U}C(p^j)	\right).
\]
Moreover, we have that
\[
|\mathfrak{F}_{C}^{U}|= \prod_{j\in\mathcal{P}^U}\left|C(p^j)\right|
\]
and, in particular,
$\mathfrak{F}^U_{C}\neq\varnothing$.
\end{theorem}

\begin{proof}Consider $f\in \mathfrak{F}^{U}_{C}$ and note that,
 for every $j\in\mathcal{P}^U$, $f(p^j)\in C(p^j)$ and $\Phi\left((f(p^j))_{j\in\mathcal{P}^U}\right)=f$. Then $\Phi$ is bijective, so that $|\mathfrak{F}_{C}^{U}|=\left|\times_{j\in\mathcal{P}^U}C(p^j)	\right|= \prod_{j\in\mathcal{P}^U}\left|C(p^j)\right|.$
 Since, for every $j\in\mathcal{P}^U$,  $C(p^j)\neq \varnothing$, it finally follows that $\mathfrak{F}^U_{C}\neq\varnothing$.
\end{proof}

\subsection{Existence of $U$-consistent resolute refinements for $U\not\leq S_h\times S_n\times \{id\}$}\label{II}
In this section we focus on the set $\mathfrak{F}^U_C$, where $U$ is a regular subgroup of $G$ not included in $S_h\times S_n\times \{id\}$ and $C$ is a $U$-consistent {\sc scc}.
We start with some crucial definitions.

Let $U\le G$ be regular such that $U\not\le S_h\times S_n\times \{id\}$, $C\in\mathfrak{C}^U$,
$(p^j)_{j\in\mathcal{P}^U}\in\mathfrak{S}(U)$ and $(\varphi_*,\psi_*,\rho_0)\in U$.
Define, for every $j\in \mathcal{P}^U_1$, the set
\begin{equation}\label{A1C}
A^1_C(p^j)=\{(y,z)\in C(p^j) \times  C(p^{j\,(\varphi_*,\psi_*,\rho_0)}): z\neq \psi_*(y)\},
\end{equation}
and, for every $j\in \mathcal{P}^U_2$, the set
\begin{equation}\label{A2C}
A^2_C(p^j)=\left\{x\in C(p^j): \psi_j(x)\neq x\right\},
\end{equation}
where $\psi_j$ is the unique element in $S_n$ such that
\begin{equation}\label{psi-j}
\mathrm{Stab}_U(p^j)\subseteq (S_h\times \{id\}\times \{id\})\cup (S_h\times \{\psi_j\}\times \{\rho_0\}).
\end{equation}
Note that that uniqueness of $\psi_j$ is guaranteed by Lemma \ref{collect}(ii).

Next if $\mathcal{P}_1^U\neq \varnothing$, then define
\[
A^1_C=\times_{j\in \mathcal{P}_1^U}A^1_C(p^j),
\]
and if $\mathcal{P}_2^U\neq \varnothing$, then define
\[
A^2_C=\times_{j\in\mathcal{P}_2^U}A^2_C(p^j).
\]

Of course, all the sets above defined depend on $U$, $(p^j)_{j\in\mathcal{P}^U}$, $(\varphi_*,\psi_*,\rho_0)$ and $C$ but we do not emphasize that dependence in the notation.
\begin{proposition}\label{fu-min-ex}
Let $U\le G$ be regular such that $U\not\le S_h\times S_n\times \{id\}$, $(p^j)_{j\in\mathcal{P}^U}\in\mathfrak{S}(U)$, $(\varphi_*,\psi_*,\rho_0)\in U$ and $C\in\mathfrak{C}^U$.
For every $j\in \mathcal{P}_1^U$, let
$(y_j,z_j)\in A^1_C(p^j)$
and, for every $j\in\mathcal{P}_2^U$, let
$x_j\in A^2_C(p^j)$.
Then there exists
a unique $f\in\mathfrak{F}^{U}_C$ such that
$f(p^j)=y_j$ and $f(p^{j\,(\varphi_*,\psi_*,\rho_0)})=z_j$ for all $j\in\mathcal{P}_1^U$, and
$f(p^j)=x_j$ for all  $j\in\mathcal{P}_2^U$.
\end{proposition}

\begin{proof}
Given $j\in\mathcal{P}_2^U$, consider the set $K^U(p^j)=\left\{\sigma\in S_n: \psi_j=\sigma \rho_0 \sigma^{-1}  \right\}$, where $\psi_j$ is defined in \eqref{psi-j}.
Since $U$ is regular, $K^U(p^j)$ is nonempty so that we can choose an element $\sigma_j$ in $K^U(p^j)$. Note that, for every $j\in\mathcal{P}_2^U$ and  $(\varphi,\psi,\rho)\in \mathrm{Stab}_U(p^j)$, we have $\psi=\sigma_j\rho\sigma_j^{-1}$.

Let us consider then $f\in\mathfrak{F}$ defined, for every $p\in\mathcal{P}$, as follows. Given $p\in \mathcal{P}$, consider the unique $j\in\mathcal{P}^U$ such that $p\in j$ and the nonempty set $U_p=\{(\varphi,\psi,\rho)\in U: p=p^{j\,(\varphi,\psi,\rho)}\}$. Pick $(\varphi,\psi,\rho)\in U_p$ and let
\begin{equation}\label{fp}
f(p)=
\left\{
\begin{array}{ll}
\psi(y_j)&\mbox{ if }j\in\mathcal{P}_1^U \mbox{ and }\rho=id\\
\vspace{-2mm}\\
\psi\psi_*^{-1}(z_j)&\mbox{ if }j\in\mathcal{P}_1^U \mbox{ and } \rho=\rho_0\\
\vspace{-2mm}\\
\psi \sigma_{j} \rho \sigma_{j}^{-1}(x_j)&\mbox{ if }j\in\mathcal{P}_2^U
\end{array}
\right.
\end{equation}
We need to prove that the value of $f(p)$ does not depend on the particular element chosen in $U_p$.
Indeed, let $(\varphi_1,\psi_1,\rho_1),(\varphi_2,\psi_2,\rho_2)\in U_p$ and recall that 
$(\varphi_2^{-1}\varphi_1,\psi_2^{-1}\psi_1,\rho_2^{-1}\rho_1)\in \mathrm{Stab}_U(p^j)$.
\begin{itemize}
\item[-] If $j\in\mathcal{P}_1^U$, then $(\varphi_2^{-1}\varphi_1,\psi_2^{-1}\psi_1,\rho_2^{-1}\rho_1)\in \mathrm{Stab}_U(p^j)$ implies
 $\rho_2=\rho_1$ and $\psi_1=\psi_2$. As a consequence, if $\rho_1=\rho_2=id$, then
$\psi_1(y_j)=\psi_2(y_j)$, while if $\rho_1=\rho_2=\rho_0$, then
$(\psi_1\psi_*^{-1})(z_j)=(\psi_2\psi_*^{-1})(z_j)$.
\item[-] If  $j\in\mathcal{P}_2^U$, then $(\varphi_2^{-1}\varphi_1,\psi_2^{-1}\psi_1,\rho_2^{-1}\rho_1)\in \mathrm{Stab}_U(p^j)$ implies
$\psi_2^{-1}\psi_1= \sigma_j \rho_2^{-1}\rho_1 \sigma_j^{-1}$, that is, $\psi_1 \sigma_j \rho_1 \sigma_j^{-1}=\psi_2 \sigma_j \rho_2 \sigma_j^{-1}$,
as $\rho=\rho^{-1}$ for all $\rho\in\Omega$. Then we get
$
\psi_1 \sigma_j \rho_1 \sigma_j^{-1}(x_j)=\psi_2 \sigma_j \rho_2 \sigma_j^{-1}(x_j).
$
\end{itemize}

We show that $f$ satisfies all the desired properties.
First of all, since $U\le G$, we have $(id,id,id)\in U$ and thus the definition of $f$ immediately implies
 $f(p^j)=y_j$ and $f(p^{j\,(\varphi_*,\psi_*,\rho_0)})=z_j$ for all $j\in\mathcal{P}_1^U$,
and $f(p^j)=x_j$ for all $j\in\mathcal{P}_2^U$.

Let us now prove that $f\in\mathfrak{F}^U$.
Consider then $p\in\mathcal{P}$ and $(\varphi,\psi,\rho)\in U$ and show that if $\rho=id$, then $f(p^{(\varphi,\psi,\rho)})=\psi f(p)$, while if $\rho=\rho_0$, then  $f(p^{(\varphi,\psi,\rho)})\neq\psi f(p)$. Let $p=p^{j\,(\varphi_1,\psi_1,\rho_1)}$ for suitable $j\in\mathcal{P}^U$ and $(\varphi_1,\psi_1,\rho_1)\in U$.
\begin{itemize}
\item[-] If $j\in\mathcal{P}_1^U$ and $\rho_1=id$, then
$
f(p)=\psi_1 (y_j).
$
By \eqref{action-e},
if $\rho=id$, then
$
f(p^{(\varphi,\psi,\rho)})=f(p^{j\,(\varphi\varphi_1,\psi\psi_1,id)})=\psi\psi_1 (y_j)=\psi f(p),
$
while if $\rho=\rho_0$, then
$
f(p^{(\varphi,\psi,\rho)})=f(p^{j\,(\varphi\varphi_1,\psi\psi_1,\rho_0)})=\psi\psi_1 \psi_*^{-1}(z_j)\neq \psi\psi_1 (y_j)=\psi f(p),
$
since $z_j\neq \psi_*( y_j)$ because $(y_j,z_j)\in A_C^1(p^j)$.
\item[-] If $j\in\mathcal{P}_1^U$ and $\rho_1=\rho_0$, then
$
f(p)=\psi_1 \psi^{-1}_*(z_j).
$
By \eqref{action-e}, if $\rho=id$, then
$
f(p^{(\varphi,\psi,\rho)})=f(p^{j\,(\varphi\varphi_1,\psi\psi_1,\rho_0)})=\psi\psi_1 \psi^{-1}_*(z_j)=\psi f(p),
$
while if $\rho=\rho_0$, then
$
f(p^{(\varphi,\psi,\rho)})=f(p^{j\,(\varphi\varphi_1,\psi\psi_1,id)})=\psi\psi_1 (y_j)\neq \psi\psi_1 \psi^{-1}_*(z_j)=\psi f(p),
$
since $z_j\neq \psi_*( y_j)$ because $(y_j,z_j)\in A_C^1(p^j)$.
\item[-] If $j\in\mathcal{P}_2^U$, then
$
f(p)=\psi_1 \sigma_{j} \rho_1 \sigma_{j}^{-1} (x_j)
$
and, by \eqref{action-e},
$
f(p^{(\varphi,\psi,\rho)})=f(p^{j\,(\varphi\varphi_1,\psi\psi_1,\rho\rho_1)})=\psi\psi_1 \sigma_{j} \rho \rho_1\sigma_{j}^{-1} (x_j).
$
As a consequence, if $\rho=id$, we get $f(p^{(\varphi,\psi,\rho)})=\psi f(p) $.
If instead $\rho=\rho_0$, we have that  $f(p^{(\varphi,\psi,\rho)})\neq\psi f(p)$ if and only if
$
\psi\psi_1 \sigma_{j} \rho_0 \rho_1 \sigma_{j}^{-1} (x_j)\neq
\psi\psi_1 \sigma_{j} \rho_1 \sigma_{j}^{-1} (x_j)
$
if and only if
$
\sigma_{j} \rho_0  \sigma_{j}^{-1} (x_j)\neq
x_j.
$
However, the last relation holds true since $\sigma_{j} \rho_0  \sigma_{j}^{-1} =\psi_j$ and $\psi_j(x_j)\neq x_j$ because $x_j\in A_C^2(p^j)$.
\end{itemize}

Let us next prove that $f\in\mathfrak{F}_C$.
Consider then $p\in\mathcal{P}$ and show that  $f(p)\in C(p)$. Let $p=p^{j\,(\varphi_1,\psi_1,\rho_1)}$ for suitable $j\in\mathcal{P}^U$ and $(\varphi_1,\psi_1,\rho_1)\in U$.
\begin{itemize}
	\item[-] If $j\in\mathcal{P}_1^U$ and $\rho_1=id$, then $f(p)=\psi_1 (y_j)$ and, by the $U$-consistency of $C$,
$\psi_1 (y_j)\in\psi_1C(p^j)=C(p^{j\,(\varphi_1,\psi_1,id)})=C(p)$.
\item[-] If $j\in\mathcal{P}_1^U$ and $\rho_1=\rho_0$, then $f(p)=\psi_1\psi_*^{-1}(z_j)$ and, by \eqref{action-e} and the $U$-consistency of $C$,
\[
\psi_1\psi_*^{-1}(z_j)\in \psi_1\psi_*^{-1}C(p^{j\,(\varphi_*,\psi_*,\rho_0)})
\]
\[
=C\left((p^{j\,(\varphi_*,\psi_*,\rho_0)})^{(\varphi_1\varphi_*^{-1},\psi_1\psi_*^{-1},id)}\right)=C(p^{j\,(\varphi_1,\psi_1,\rho_0)})=C(p).
\]
\item[-] If $j\in\mathcal{P}_2^U$ and $\rho_1=id$, then $f(p)=\psi_1(x_j)$ and, by the $U$-consistency of  $C$,
$\psi_1(x_j)\in \psi_1 C(p^j)
=C(p^{j\,(\varphi_1,\psi_1,id)})=C(p)$.
\item[-] If $j\in\mathcal{P}_2^U$ and $\rho_1=\rho_0$, then let
 $(\varphi_2,\psi_2,\rho_0)\in U$ be such that $p^{j\,(\varphi_2,\psi_2,\rho_0)}=p^j$. By \eqref{action-e}, we have
\begin{equation}\label{useful-1}
p=p^{j\,(\varphi_1,\psi_1,\rho_0)}=(p^{j\,(\varphi_2,\psi_2,\rho_0)})^{(\varphi_1\varphi_2^{-1},\psi_1\psi_2^{-1},id)}=p^{j\,(\varphi_1\varphi_2^{-1},\psi_1\psi_2^{-1},id)}.
\end{equation}
Thus, $f(p)=\psi_1\psi_2^{-1}(x_j)$ and, by the $U$-consistency of  $C$,
$\psi_1\psi_2^{-1}(x_j)\in \psi_1\psi_2^{-1} C(p^j)
=C(p^{j\,(\varphi_1\varphi_2^{-1},\psi_1\psi_2^{-1},id)})=C(p)$.
\end{itemize}

Finally, in order to prove uniqueness, let $f'\in\mathfrak{F}^U_C$ such that
 $f'(p^j)=y_j$ and $f'(p^{j\,(\varphi_*,\psi_*,\rho_0)})=z_j$ for all $j\in\mathcal{P}_1^U$,
and $f'(p^j)=x_j$ for all $j\in\mathcal{P}_2^U$. Then $f,f'\in \mathfrak{C}^U$ realize $f(p^j)=f'(p^j)$ for all $j\in \mathcal{P}^U$ and $f(p^{j\,(\varphi_*,\psi_*,\rho_0)})=f'(p^{j\,(\varphi_*,\psi_*,\rho_0)})$ for all $j\in \mathcal{P}_1^U$. Hence, the thesis follows from Proposition \ref{rappresentanti2}.
\end{proof}

Let $U\le G$ be regular such that $U\not\le S_h\times S_n\times \{id\}$, $(p^j)_{j\in\mathcal{P}^U}\in\mathfrak{S}(U)$, $(\varphi_*,\psi_*,\rho_0)\in U$ and $C\in\mathfrak{C}^U$.
\begin{itemize}
\item[-] If $\mathcal{P}_2^U=\varnothing$, then let $\Psi_1: A^1_C\to \mathfrak{F}^U_C$ be the function which associates with every  $(y_j,z_j)_{j\in\mathcal{P}_1^U}\in A^1_C$, the unique $f\in\mathfrak{F}^U_C$ defined in Proposition \ref{fu-min-ex}.
	\item[-] If $\mathcal{P}_1^U=\varnothing$, then let
$\Psi_2: A^2_C\to \mathfrak{F}^U_C$ be the function which associates with every  $(x_j)_{j\in \mathcal{P}_2^U}\in A^2_C$, the unique $f\in\mathfrak{F}^U_C$ defined in Proposition \ref{fu-min-ex}.
\item[-] If $\mathcal{P}_1^U\neq \varnothing$ and $\mathcal{P}_2^U\neq \varnothing$, then let
$\Psi_3:A^1_C\times A^2_C \to \mathfrak{F}^U_C$ be the function
which associates with every  $((y_j,z_j)_{j\in\mathcal{P}_1^U},(x_j)_{j\in \mathcal{P}_2^U})\in A^1_C\times A^2_C$, the unique $f\in\mathfrak{F}^U_C$ defined in Proposition \ref{fu-min-ex}.
\end{itemize}
Of course, $\Psi_1$, $\Psi_2$ and $\Psi_3$ depend on $U$, $(p^j)_{j\in\mathcal{P}^U}$, $(\varphi_*,\psi_*,\rho_0)$ and $C$ but we do not emphasize that dependence in the notation. Note also that $\Psi_1$, $\Psi_2$ and $\Psi_3$  are injective.

\begin{theorem}\label{fu-min-count2}
Let $U\le G$ be regular such that $U\not\le S_h\times S_n\times \{id\}$, $(p^j)_{j\in\mathcal{P}^U}\in\mathfrak{S}(U)$, $(\varphi_*,\psi_*,\rho_0)\in U$ and $C\in\mathfrak{C}^U$.
Then
\[
\mathfrak{F}^U_C=
\left\{
\begin{array}{ll}
\Psi_1(A^1_C) & \mbox{if }\mathcal{P}_2^U=\varnothing\\
\Psi_2(A^2_C) & \mbox{if }\mathcal{P}_1^U=\varnothing\\
\Psi_3(A^1_C\times A^2_C)& \mbox{if }\mathcal{P}_1^U\neq\varnothing\mbox{ and }\mathcal{P}_2^U\neq\varnothing\\
\end{array}
\right.
\]
Moreover, we have that
\[
|\mathfrak{F}^U_C|=
\left\{
\begin{array}{ll}
|A^1_C| & \mbox{if }\mathcal{P}_2^U=\varnothing\\
|A^2_C| & \mbox{if }\mathcal{P}_1^U=\varnothing\\
|A^1_C|\cdot |A^2_C|& \mbox{if }\mathcal{P}_1^U\neq\varnothing\mbox{ and }\mathcal{P}_2^U\neq\varnothing\\
\end{array}
\right.
\]
and  $\mathfrak{F}^U_C\neq\varnothing$.
\end{theorem}

\begin{proof}Assume first that $\mathcal{P}_1^U$ and  $\mathcal{P}_2^U$ are both nonempty.
Consider $f\in \mathfrak{F}^{U}_{C}$ and note that
\[
\left((f(p^j),f(p^{j\,(\varphi_*,\psi_*,\rho_0)}))_{j\in\mathcal{P}_1^U},(f(p^j))_{j\in \mathcal{P}_2^U}\right)\in A^1_{C}\times A^2_{C},
\]
and
\[
\Psi_3\left((f(p^j),f(p^{j\,(\varphi_*,\psi_*,\rho_0)}))_{j\in\mathcal{P}_1^U},(f(p^j))_{j\in \mathcal{P}_2^U}\right)=f.
\]
Then $\Psi_3$ is bijective, so that $|\mathfrak{F}^U_C|=|A^1_C\times A^2_C|=|A^1_C|\cdot| A^2_C|$. We complete the proof showing that,
for every $j\in \mathcal{P}^U_1$, $A^1_C(p^j)\neq \varnothing$ and that, for every $j\in \mathcal{P}^U_2$, $A^2_C(p^j)\neq \varnothing$.
The fact that $A^1_C(p^j)$ has at least one element for all $j\in \mathcal{P}^U_1$ is an immediate consequence of the $U$-consistency of $C$. Assume now that  there exists $j\in \mathcal{P}^U_2$ such that $A^2_C(p^j)= \varnothing$ and consider $\psi_j$ as defined in \eqref{psi-j}.
Then, for every $x\in C(p^j)$, we have that $\psi_j(x)=x$. On the other hand, being  $\psi_j$ a conjugate of $\rho_0$, it has the same number of fixed points of $\rho_0$. Thus, if $n$ is even, then $\psi_j$ has no fixed point and  so $C(p^j)= \varnothing$, a contradiction. If instead $n$ is odd, we have that $\psi_j$ has a unique fixed point $x_0$ and so $C(p^j)=\{x_0\}$. Pick  $(\varphi_1,\psi_1,\rho_0)\in \mathrm{Stab}_U(p^j).$ By the regularity of $U$, we get $\psi_1=\psi_j$  and thus $\psi_1(x_0)=x_0.$ It follows that $\psi_1^{-1}C(p^j)=C(p^j).$ Now, by \eqref{action-e} and the $U$-consistency of $C$, we finally deduce that
\[
C(p^{j\,(\varphi_*,\psi_*,\rho_0)})
=C\left((p^{j\,(\varphi_1,\psi_1,\rho_0)})^{(\varphi_*\varphi_1^{-1},\psi_*\psi_1^{-1}, id)}\right)
\]
\[
=C(p^{j	\,(\varphi_*\varphi_1^{-1},\psi_*\psi_1^{-1}, id)})
=\psi_*\psi_1^{-1}C(p^j)=
\psi_*C(p^j),
\]
which contradicts \eqref{sccU2}.

The case $\mathcal{P}_1^U=\varnothing$ and the case $\mathcal{P}_2^U=\varnothing$ are similar and then omitted.
\end{proof}

\section{Some applications}\label{appl}

In this section we mainly apply the general theory about the concept of consistency to study the properties of anonymity and neutrality with respect to partitions as well as the immunity to the reversal bias. 
In particular, we   describe some concrete situations involving the classical {\sc scc}s considered in Section \ref{B-C-K-rb}.
In what follows we denote by $\mathfrak{C}^*$ the set $\{Par,Bor,Cop,Kem,Min\}$.

\subsection{Proof of Theorem \ref{super-super}}

We are going to prove Theorem \ref{super-super} by proving Theorem \ref{super-condensed} below. Indeed, on the basis of the notation introduced along the paper and Proposition \ref{U,V-corr}, Theorem  \ref{super-condensed} is nothing but a rephrase of Theorems \ref{super} and \ref{super-super}. More precisely, statement $(i)$ refers to Theorem \ref{super}, while statement $(ii)$ refers to 
Theorem \ref{super-super}.
Since statement $(i)$ has been already proved in Section \ref{main}, we are left with proving statement $(ii)$ only.

\begin{theorem}\label{super-condensed}
Let $Y=\{Y_j\}_{j=1}^s$ be a partition of $H$, $Z=\{Z_k\}_{k=1}^t$ be a partition of $N$ with $|Z_{k^*}|=\max\{ |Z_k| \}_{k=1}^t$, and $R\le \Omega$.
\begin{itemize}
	\item[(i)]If  $\mathfrak{F}^{V(Y)\times W(Z)\times R}_{Par}\neq \varnothing$, then \eqref{reg-eq}.
	\item[(ii)]If  $C\in\mathfrak{C}^{V(Y)\times W(Z)\times R}$ and \eqref{reg-eq}, then
	$\mathfrak{F}^{V(Y)\times W(Z)\times R}_{C}\neq\varnothing$.
\end{itemize}
\end{theorem}

\begin{proof}
$(ii)$  Let $U=V(Y)\times W(Z)\times R$ and $C\in\mathfrak{C}^{U}$. By Theorem \ref{regular}, \eqref{reg-eq} implies that $U$ is regular. If 
$R=\{id\}$, then we apply Theorem \ref{fu-min-2} and we get $\mathfrak{F}^{U}_{C}\neq\varnothing$. If instead $R=\Omega$, then we apply Theorem  \ref{fu-min-count2} and we get again $\mathfrak{F}^{U}_{C}\neq\varnothing$.
\end{proof}

 \subsection{Five individuals and three alternatives}\label{example53}

Consider five individuals ($h=5$) and three alternatives ($n=3$) so that $G=S_5\times S_3\times \Omega$.  Since $\gcd(5,3!)=1$, Theorem \ref{regular} guarantees that $G$ is regular. Thus, by Theorem \ref{fu-min-count2},  for  every  $C\in \mathfrak{C}^{G} $, we have $\mathfrak{F}^G_{C}\neq \varnothing$ and the elements in $\mathfrak{F}^G_{C}$ can be explicitly build and  count. Here we determine $\mathfrak{F}^G_{C}$ for $C\in\mathfrak{C}^*$. Observe that,  being $n=3$, Proposition \ref{moulin2} guarantees that $\mathfrak{C}^*\subseteq \mathfrak{C}_{Par}^{G} $.

In order to start the concrete construction of the elements in $\mathfrak{F}^G_{C}$, we first need a system of representatives of the $G$-orbits. We choose the system $p^1,\ldots,p^{26}$ built in
Bubboloni and Gori (2015, Section 7.2) and, for every $j\in\{1,\ldots,26\}$, we denote the orbit of $p^j$ by $j$. Thus, $\mathcal{P}^G=\{1,\ldots,26\}$ and
a simple but tedious computation shows that
\[
\mathcal{P}^G_1=\{2,4,5,7,8,10,11,12,14,15,16,18,21,22,23,24\},
\]
\[
\mathcal{P}^G_2=\{1,3,6,9,13,17,19,20,25,26\}.
\]
Next, we choose $(\varphi_*,\psi_*,\rho_0)=(id,id,\rho_0)$ 
and, for every $C\in\mathfrak{C}^*$, we compute $C(p^j)$ for all $j\in \mathcal{P}^G$, and $C(p^{j\,(id,id,\rho_0)})$
for all $j\in \mathcal{P}^G_1$. Doing that, we find out that $Kem(p^j)=Min(p^j)$  for all $j\in \mathcal{P}^G$, as well as $Kem(p^{j\,(id,id,\rho_0)})=Min(p^{j\,(id,id,\rho_0)})$ for all $j\in \mathcal{P}^G_1$. Thus, Proposition \ref{rappresentanti2}  gives $Kem=Min$.
Next we compute, for every $j\in \mathcal{P}_1^G$, the set $A^1_{C}(p^j)$ defined in \eqref{A1C} and,  for every $j\in \mathcal{P}_2^G$, the set $A^2_{C}(p^j)$ defined in  \eqref{A2C}.
 Those computations are summarized in Tables 1 and 2, where
\begin{equation}\label{Delta}
\Delta=\{1,2,3\}_*^2=\{(1,2),(2,1),(1,3),(3,1),(2,3),(3,2)\}.
\end{equation}
 From those tables, by Theorem \ref{fu-min-count2}, we immediately get every element in $\mathfrak{F}^G_{C}$ for all $C\in\mathfrak{C}^*$. Indeed, once decided, for every $j\in\{1,\ldots,26\}$, one of the entries corresponding to $p^j$, we exactly identify an element in $\mathfrak{F}^G_{C}$ just using  the definition  given in \eqref{fp}.
 
 In particular, we deduce
\[
\mathfrak{F}^G_{Kem}=\mathfrak{F}^G_{Min},    \quad \mathfrak{F}^G_{Kem}\subsetneq  \mathfrak{F}^G_{Bor}\subsetneq \mathfrak{F}^G_{Par}, \quad \mathfrak{F}^G_{Kem}\subsetneq   \mathfrak{F}^G_{Cop}\subsetneq \mathfrak{F}^G_{Par},\quad \mathfrak{F}^G_{Bor}\not\subseteq   \mathfrak{F}^G_{Cop},\quad \mathfrak{F}^G_{Cop}\not\subseteq   \mathfrak{F}^G_{Bor}
\]
and
 \[|\mathfrak{F}^G_{Par}|=2^{20} 3^{14},\quad |\mathfrak{F}^G_{Bor}|=8,\quad |\mathfrak{F}^G_{Cop}|=4,\quad |\mathfrak{F}^G_{Kem}|=|\mathfrak{F}^G_{Min}|=2.\]
\begin{table}[ht]
\centering
\begin{tabular}{l|c|c|c|c|c|llllllllllllllllllll}
  & $A^1_{Par}(p^j)$ &  $A^1_{Bor}(p^j)$ & $A^1_{Cop}(p^j)$, $A^1_{Kem}(p^j)$, $A^1_{Min}(p^j)$ \\
	\hline
$p^2$&    (1,3), (2,3) &  (1,3) & (1,3) \\
$p^4$&    (1,3), (2,1), (2,3)	&  (1,3) & (1,3)  \\
$p^5$	&    (1,3),  (2,3) & (1,3) & (1,3)  \\
$p^7$	&    (1,3), (2,1), (2,3)& (2,3) & (1,3) \\
$p^8$	&    $\Delta$& (1,3) & (1,3)   \\
$p^{10}$	&     (1,3), (2,1), (2,3)	& (1,3) & (1,3)  \\
$p^{11}$	&   $\Delta$ 	& (1,3) & (1,3) \\
$p^{12}$	&   $\Delta$	& (1,3), (2,3) & (1,3)  \\
$p^{14}$ &    $\Delta$ & (2,3) & (2,3)  \\
$p^{15}$	&    (1,2), (1,3),  (2,3) & (1,3) & (1,3)  \\
$p^{16}$	&    $\Delta$& (2,3) & (2,3)  \\
$p^{18}$&    $\Delta$	& (2,1), (2,3) & (2,1) \\
$p^{21}$&    $\Delta$	& (2,3) & (2,3)   \\
$p^{22}$	&    $\Delta$ & (1,3) & (1,3) \\
$p^{23}$&    $\Delta$	& (1,3) & (1,3)  \\
$p^{24}$ &   $\Delta$& (1,3) & (1,3)  \\
\hline
\end{tabular}
\caption{Computation of $A^1_{C}(p^j)$ with $C\in\mathfrak{C}^*$ and $j\in\mathcal{P}^G_1$.}
\end{table}

\begin{table}[ht]
\centering
\begin{tabular}{l|c|c|c|c|c|lllllllllllllllllllll}
  & $A^2_{Par}(p^j)$ &  $A^2_{Bor}(p^j)$ & $A^2_{Cop}(p^j)$ & $A^2_{Kem}(p^j)$, $A^2_{Min}(p^j)$ \\
	\hline
$p^1$& 1 &  1 & 1&1\\
$p^3$& 1,\;3	&   1 & 1&1\\
$p^6$& 1,\;3	& 1 & 1&1 \\
$p^9$& 1	& 1 & 1&1 \\
$p^{13}$& 1,\;3	&1 & 1&1\\
$p^{17}$& 1,\;2	&2 & 2&2\\
$p^{19}$& 1,\;2&	1 & 1,\;2&1\\
$p^{20}$& 1,\;3	&1 & 1&1\\
$p^{25}$& 1,\;3  &1,\;3 & 1,\;3&1,\;3\\
$p^{26}$& 2,\;3	&2 & 2&2\\
\hline
\end{tabular}
\caption{Computation of $A^2_{C}(p^j)$ with $C\in \mathfrak{C}^*$ and $j\in\mathcal{P}^G_2$.}
\end{table}
Note that the two $G$-consistent refinements of $Kem$ depend only on which alternative between 1 and 3 is associated with the following preference profile
\[
p^{25}=
\left[
\begin{array}{ccccc}
 1 & 1 & 3 & 2 & 3\\
 2 & 2 & 2 & 3 & 1\\
 3 & 3 & 1 & 1 & 2\\
\end{array}
\right]
\]
The choice of $3$ might be considered more appropriate since the majority of individuals prefer 3 to 1.

\subsection{Committees with a distinguished individual}

We have observed in Section \ref{main} that, if individual $i$ has a special role in the committee, then $Y=\{\{i\},H\setminus\{i\}\}$ 
is a natural partition of individuals  to deal with. 
Let $C\in \mathfrak{C}$ and consider $C^{i}, C_i\in \mathfrak{C}$ defined, for every $p\in\mathcal{P}$, by
\[
\begin{array}{l}
C^{i}(p)=\{x\in C(p): \forall y\in C(p),\, x\succeq_{p_i} y \},\\
\vspace{-2mm}\\
C_i(p)=\{x\in C(p): \forall y\in C(p),\, y\succeq_{p_i} x \}.
\end{array}
\]
Of course, $C^i$ and $C_i$ are resolute refinements of $C$.
Note that $C^i$ is consistent with interpreting individual $i$ as the president of the committee who has the power to break ties according to his/her own preferences, while $C_i$ does not seem to have any natural interpretation.

If $C$ is $Y$-anonymous [neutral], it is easily checked that $C^i$ and $C_i$ are both $Y$-anonymous [neutral]. If instead $C$ is immune to the reversal bias, it is not generally guaranteed that $C^i$ and $C_i$ are immune to the reversal bias too. Indeed, consider, for instance, $(h,n)=(5,4)$ and the Minimax {\sc scc}. Recall that, by Proposition \ref{minimax}, $Min$ is immune to the reversal bias. Given now $p, \hat{p}\in \mathcal{P}$  defined by
 \[
p=
\left[
\begin{array}{ccccc}
4&4&4&2&1\\
2&3&1&3&2\\
3&1&2&1& 3\\
1&2&3&4&4
\end{array}
\right],\quad 
\hat{p}=
\left[
\begin{array}{ccccc}
1&4&4&2&2\\
2&3&1&3&2\\
3&1&2&1&3\\
4&2&3&4&1
\end{array}
\right]
\]
we have that $Min(p)=Min(\hat{p})=\{4\}$ and  $Min(p^{(id,id,\rho_0)})=Min(\hat{p}^{(id,id,\rho_0)})=\{1,3,4\}.$ Then we have
\[
Min^5(p)=Min^5(p^{(id,id,\rho_0)})=Min_5(\hat{p})=Min_5(\hat{p}^{(id,id,\rho_0)})=\{4\},
\] 
so that $Min^5$ and $Min_5$ suffers the reversal bias. However, note that, due to Theorem \ref{super-super}, we know that $Min$ surely admits 
a resolute refinement which is $\{\{5\},H\setminus\{5\}\}$-anonymous, neutral and immune to the reversal bias.


Remarkably, if $C$ satisfies a suitable stronger version of immunity to the reversal bias, $C^i$ and $C_i$ can be proved to be immune to the reversal bias. Accordingly to Bubboloni and Gori (2016) we say that $C\in \mathfrak{C}$ is immune to the reversal bias of type $3$ if, for every $p\in \mathcal{P}$, $C(p)\cap C(p^{(id,id,\rho_0)})\neq\varnothing$  implies $C(p)= N$. Note that if $C$ is resolute then the definitions of immunity to the reversal bias and immunity to the reversal bias of type $3$ coincide. 

\begin{proposition}\label{type3} Let $i\in H$ and $C\in \mathfrak{C}$ be immune to the reversal bias of type $3$. Then $C^i$ and $C_i$ are  immune to the reversal bias.
\end{proposition}

\begin{proof} 
Assume by contradiction that there exists $p\in \mathcal{P}$ and $x^*\in N$ such that
\[
C^i(p)=C^i(p^{(id,id,\rho_0)})=\{x^*\}.
\]
Then, in particular, $x^*\in C(p)\cap C(p^{(id,id,\rho_0)})$ and, since $C$ is immune to the reversal bias of type $3$, we have $C(p)=C(p^{(id,id,\rho_0)})= N$. Being $x^*\in C^i(p)$ we then have $x^*\succeq_{p_i} y$ for all $y\in N.$ On the other hand, being  $x^*\in C^i(p^{(id,id,\rho_0)})$ we also have $x^*\succeq_{p_i\rho_0} y$, that is, $y\succeq_{p_i}x^*$ for all $y\in N$. Since $p_i$ is antisymmetric, we then get that $x^*$ is the only element in $N$, against $n\geq 2.$
An analogous argument works for $C_i$.
\end{proof}

\begin{corollary}\label{bor3} Let $i\in H$ and  $Y=\{\{i\},H\setminus\{i\}\}$. Then $Bor^i$, $Bor_i$, $Cop^i$ and $Cop_i$ are efficient, $Y$-anonymous, neutral and immune to the reversal bias.
\end{corollary}
\begin{proof} We know that $Bor$ and $Cop$ are efficient, anonymous and neutral. By Proposition 3 in Bubboloni and Gori (2016), we have that $Bor$ and $Cop$ are immune to the reversal bias of type $3$. Thus, the thesis follows applying Proposition \ref{type3} and recalling that any refinement of an efficient {\sc scc} is efficient.
\end{proof}

\subsubsection{Three individuals and three alternatives}\label{example33} 

Consider now three individuals ($h=3$) and three alternatives ($n=3$) so that $G=S_3\times S_3\times \Omega$.  Since $\gcd(3,3!)\neq 1$, $G$ is not regular and, by Theorem \ref{super}, there exists no resolute, efficient, anonymous and neutral {\sc scc}. Thus, $\mathfrak{F}_C^G=\varnothing$ for all $C\in\mathfrak{C}^*$. 

Consider then the partition $Y=\{\{1,2\},\{3\}\}$ of $H$ distinguishing individual $3$ and the partition $Z=\{N\}$ of $N$ and define 
$U=V(Y)\times W(Z)\times \Omega$. By Theorem \ref{regular}, $U$ is regular so that, by Theorem \ref{fu-min-count2},  for  every  $C\in \mathfrak{C}^{U} $, we have that $\mathfrak{F}^U_{C}\neq \varnothing$ and all the elements in $\mathfrak{F}^U_{C}$ can be explicitly built and counted.
Here we determine $\mathfrak{F}^U_{C}$ for all $C\in\mathfrak{C}^*$. Observe that,  being $n=3$, Proposition \ref{moulin2} guarantees that $\mathfrak{C}^*\subseteq \mathfrak{C}^{G} $ so that $\mathfrak{C}^*\subseteq \mathfrak{C}^{U}$.

As a system of representatives of the $U$-orbits, we consider the system $p^1,\ldots,p^{13}$ built in
Bubboloni and Gori (2015, Section 7.1) and, for every $j\in\{1,\ldots,13\}$, we denote the orbit of $p^j$ by $j$. Thus, $\mathcal{P}^U=\{1,\ldots,13\}$ and a simple computation shows that
\[
\mathcal{P}^U_1=\{3,4,8,9,10,11,12,13\}, \quad \mathcal{P}^U_2=\{1,2,5,6,7\}.
\]
Next, we choose $(\varphi_*,\psi_*,\rho_0)=(id,id,\rho_0)$ 
and, for every $C\in\mathfrak{C}^*$, we compute $C(p^j)$ for all $j\in \mathcal{P}^U$, and $C(p^{j\,(id,id,\rho_0)})$
for all $j\in \mathcal{P}^U_1$. Doing that, we find out that $Cop(p^j)=Kem(p^j)=Min(p^j)$  for all $j\in \mathcal{P}^U$, as well as $Cop(p^{j\,(id,id,\rho_0)})=Kem(p^{j\,(id,id,\rho_0)})=Min(p^{j\,(id,id,\rho_0)})$ for all $j\in \mathcal{P}^U_1$. Thus, by Proposition \ref{rappresentanti2}, we have  $Cop=Kem=Min$. In particular, those {\sc scc}s admit the same resolute refinements.

Next we compute, for every $j\in \mathcal{P}_1^G$, the set $A^1_{C}(p^j)$ defined in \eqref{A1C} and,  for every $j\in \mathcal{P}_2^U$, the set $A^2_{C}(p^j)$ defined in  \eqref{A2C}. Note that here $A^2_{C}(p^j)=C(p^j)\setminus \{2\}$  because $\psi_j=\rho_0=(13)$ for all $j\in  \mathcal{P}_2^U.$
Those computations are summarized in Tables 3 and 4, where $\Delta$ is defined by \eqref{Delta}.
From those tables, by Theorem \ref{fu-min-count2}, we immediately get each element in $\mathfrak{F}^U_{C}$ for all $C\in \mathfrak{C}^*$ as described in Section \ref{example53}. 

 In particular, we deduce
\[
\mathfrak{F}^U_{Cop}=\mathfrak{F}^U_{Kem}=\mathfrak{F}^U_{Min},    \quad \mathfrak{F}^G_{Cop}\subsetneq  \mathfrak{F}^G_{Bor}\subsetneq \mathfrak{F}^G_{Par}.
\]
and
 \[|\mathfrak{F}^U_{Par}|=2^{10} 3^{5},\quad |\mathfrak{F}^G_{Bor}|=8,\quad|\mathfrak{F}^U_{Cop}|= |\mathfrak{F}^U_{Kem}|=|\mathfrak{F}^U_{Min}|=2.\]

Note that the two $U$-consistent refinements of $Cop$ depend only on which alternative between 1 and 3 is associated with the preference profile
\[
p^6=
\left[
\begin{array}{ccc}
 2 & 3 & 1  \\
 3 & 1 &  2 \\
 1 & 2 &  3 \\
\end{array}
\right]
\]
Observe that, $p^6$ gives rise to the classical Condorcet cycle. 
By Corollary \ref{bor3}, $Cop^3$ and $Cop_3$ belong to $\mathfrak{F}^U_{Cop}$. It is also clear that $Cop^3\neq Cop_3$, because a direct computation shows that $Cop^3(p^6)=1$ and  $Cop_3(p^6)=3.$ Thus, we have that
$\mathfrak{F}^U_{Cop}=\{Cop^3, Cop_3\}$.

\begin{table}[ht]
\centering
\begin{tabular}{l|c|c|c|c|c|llllllllllllllllllll}
  & $A^1_{Par}(p^j)$ &  $A^1_{Bor}(p^j)$ & $A^1_{Cop}(p^j)$, $A^1_{Kem}(p^j)$, $A^1_{Min}(p^j)$ \\
	\hline
$p^3$&    (1,2), (1,3), (3,2) & (1,2),  (3,2) & (3,2) \\
$p^4$&    (1,2),  (1,3)	&  (1,2) & (1,2)  \\
$p^8$	&    (1,3),  (2,3) & (1,3) & (1,3)  \\
$p^9$	&    (1,3),(1,2), (2,1), (2,3)& (1,3),(2,3) & (1,3) \\
$p^{10}$	& $\Delta$    	& (3,2) & (3,2)  \\
$p^{11}$	&   $\Delta$ 	& (1,2) & (1,2) \\
$p^{12}$	&   $\Delta$	& (1,3) & (1,3)  \\
$p^{13}$ &  (1,3), (2,1), (2,3)   & (2,3) & (2,3)  \\
\hline
\end{tabular}
\caption{Computation of $A^1_{C}(p^j)$ with $C\in\mathfrak{C}^*$ and $j\in\mathcal{P}^U_1$.}
\end{table}

\begin{table}[ht]
\centering
\begin{tabular}{l|c|c|c|c|c|lllllllllllllllllll}
  & $A^2_{Par}(p^j)$ &  $A^2_{Bor}(p^j)$, $A^2_{Cop}(p^j)$,  $A^2_{Kem}(p^j)$, $A^2_{Min}(p^j)$ \\
	\hline
$p^1$& 1 &  1 \\
$p^2$& 1,\;3	&   3 \\
$p^5$& 1,\;3	& 1  \\
$p^6$& 1,\;3	& 1,\;3  \\
$p^{7}$& 1  &1 \\
\hline
\end{tabular}
\caption{Computation of $A^2_{C}(p^j)$ with $C\in\mathfrak{C}^*$ and $j\in\mathcal{P}^U_2$.}
\end{table}

\appendix

\section{Appendix}

Given $U$ be a regular subgroup of $G$, we characterize here when $\mathcal{P}_1^U\neq\varnothing$ or $\mathcal{P}_2^U\neq\varnothing$.
Note that, obviously, if $U$ is a regular subgroup of $G$ with  $U\leq S_h\times S_n\times \{id\}$, then $\mathcal{P}_2^U=\varnothing$ and $\mathcal{P}^U=\mathcal{P}_1^U\neq \varnothing.$ An example of such a kind of subgroup is $S_h\times\{id\}\times\{id\}.$ On the other hand, there surely exist regular subgroups of $G$ not contained in $S_h\times S_n\times \{id\}$ like, for instance, $\{id\}\times\{id\}\times \Omega$ and $\langle (\varphi,\rho_0,\rho_0)\rangle$, where $\varphi\in S_h$ is fixed. The following proposition is about subgroups of that type.

\begin{proposition}\label{p1p2}
Let $U\le G$ be regular and such that $U\not\le S_h\times S_n\times \{id\}$. 
\begin{itemize}
	\item[(i)] If $n= 2$, then $\mathcal{P}^U_1\neq \varnothing$ if and only if there exists a partition $Y=\{Y_1,Y_2\}$ of $H$ such that $\Gamma \cap V(Y)=\varnothing$, where
	$\Gamma=\{\varphi\in S_h: (\varphi,\rho_0,\rho_0)\in U\}$ and $V(Y)=\left\{\varphi\in S_h : \varphi(Y_1)=Y_1\right\}$.
	\item[(ii)] If $n\ge 3$, then $\mathcal{P}_1^U \neq\varnothing$.
	\item[(iii)] $\mathcal{P}_2^U \neq \varnothing$ if and only if there exists $(\varphi,\psi,\rho_0)\in U$ such that $\psi$ is a conjugate of $\rho_0$.
\end{itemize}
\end{proposition}
\begin{proof} $(i)$ First of all, note that $n=2$ gives $U\leq S_h\times S_2\times \Omega$, so that, for every $p\in\mathcal{P}$ and  $i\in H$, $p_i\in S_2=\{id,\rho_0\}.$
Since the only conjugate of $\rho_0$ in $S_2$ is $\rho_0$,
$U$ is regular if and only if, for every $p\in\mathcal{P}$, $\mathrm{Stab}_U(p)\subseteq \left(S_h\times \{id\}\times \{id\}\right)\cup \left( S_h\times \{\rho_0\}\times \{\rho_0\}\right).$

Assume that $Y=\{Y_1,Y_2\}$ is a partition of $H$ such that $\Gamma \cap V(Y)=\varnothing$ and prove that $\mathcal{P}^U_1\neq \varnothing$.
Consider $p\in\mathcal{P}$  defined by $p_i=id$ for all $i\in Y_1$, and  $p_i=\rho_0$ for all $i\in Y_2$. We prove that
$\mathrm{Stab}_U(p)\leq S_h\times \{id\}\times \{id\}$. Indeed, assume by contradiction that there exists $\varphi\in S_h$ such that $(\varphi,\rho_0,\rho_0)\in \mathrm{Stab}_U(p)$.
Then $\varphi\in \Gamma$, so that  $\Gamma \cap V(Y)=\varnothing$ guarantees the existence
of $i_1\in Y_1$ such that $\varphi(i_1)\in Y_2.$ Moreover, by $(\varphi,\rho_0,\rho_0)\in \mathrm{Stab}_U(p),$ we have $p_{\varphi(i)}=\rho_0 p_i\rho_0$ for all $i\in H.$ But, since $S_2$ is abelian and $\rho_0^2=id$, that gives $p_{\varphi(i)}= p_i$ for all $i\in H.$ In particular, $\rho_0 =p_{\varphi(i_1)}= p_{i_1}=id,$ a contradiction.

Assume now that, for every partition $Y=\{Y_1,Y_2\}$ of $H$, we have that $\Gamma \cap V(Y)\neq \varnothing$ and prove that $\mathcal{P}^U_1=\varnothing$. We need to show that, for every $p\in\mathcal{P}$, there exists $\varphi\in S_h$ such that $(\varphi,\rho_0,\rho_0)\in \mathrm{Stab}_U(p)$. First of all, note that $\Gamma \cap V(Y)\neq \varnothing$ implies $\Gamma\neq \varnothing$.
If $p\in\mathcal{P}$ is constant, that is $p_i=p_j$ for all $i,j\in H$, choose any $\varphi\in \Gamma$ and note that, being $S_2$ abelian and $\rho_0^2=id$, we have $p^{(\varphi,\rho_0,\rho_0)}_i=\rho_0p_{\varphi^{-1}(i)}\rho_0=p_{\varphi^{-1}(i)}=p_i.$ If instead $p\in\mathcal{P}$ is not constant, consider the partition $Y=\{Y_1,Y_2\}$ of $H$, where $Y_1=\{i\in H: p_i=id\}$ and $Y_2=\{i\in H: p_i=\rho_0\}$. By assumption, there exists $\varphi\in S_h$ such that $(\varphi,\rho_0,\rho_0)\in U$ and $\varphi(Y_1)=Y_1$, so that $\varphi(Y_2)=Y_2$ too. It follows that, for every $i\in H,$ we have $p^{(\varphi,\rho_0,\rho_0)}_{\varphi(i)}=\rho_0p_{i}\rho_0=p_{i}=p_{\varphi(i)},$
because $i$ and $\varphi(i)$  both belong either to $Y_1$ or to $Y_2.$ Thus, $(\varphi,\rho_0,\rho_0)\in \mathrm{Stab}_U(p)$.

$(ii)$ We have to exhibit $p\in\mathcal{P}$ such that $\mathrm{Stab}_U(p)\leq S_h\times\{id\}\times\{id\}$. Consider then $p\in\mathcal{P}$ defined by $p_1=id$ and $p_i=(12),$ for all $i\in H\setminus\{1\}$. In particular, we have $p_1(1)=1$ and $p_i(1)=2$ for all  $i\in H\setminus\{1\}$, so that in $p$ the individuals rank first different alternatives.
Given now $(\varphi,\psi,\rho_0)\in U$, we have that the preference profile $p^{(\varphi,\psi,\rho_0)}$ admits one component equal to $\psi\rho_0$ and $h-1$ components equal to $\psi(12)\rho_0$. Observe that $\psi\rho_0(1)=\psi(n)$ as well as, being $n\geq3$, $\psi(12)\rho_0(1)=\psi(12)(n)=\psi(n).$ Thus, in $p^{(\varphi,\psi,\rho_0)}$ each individual ranks first the same alternative $\psi(n)$, which implies $p^{(\varphi,\psi,\rho_0)}\neq p.$

$(iii)$ If there exists $(\varphi,\psi,\rho_0)\in U$ such that $\psi$ is a conjugate of $\rho_0$, then there exist $\sigma\in S_n$ such that $\psi=\sigma \rho_0\sigma^{-1}$. Consider  $p\in\mathcal{P}$ defined by $p_i=\sigma$ for all $i\in H$. Then, for every $i\in H$, we have that $p^{(\varphi,\psi,\rho_0)}_i=\sigma \rho_0 \sigma^{-1}\sigma \rho_0=\sigma=p_i$ so that $(\varphi,\psi,\rho_0)\in \mathrm{Stab}_U(p)$ and $\mathcal{P}_2^U \neq \varnothing$.
On the other hand, if $\mathcal{P}_2^U \neq \varnothing,$ then there exists $p\in\mathcal{P}$ and $(\varphi,\psi,\rho_0)\in \mathrm{Stab}_U(p)$. Thus, by the regularity of $U$, $\psi$ is a conjugate of $\rho_0$.
\end{proof}

We emphasize that all the situations described in Proposition \ref{p1p2} can really occur.
Indeed, if $U=\{id\}\times\{id\}\times \Omega$, then $\mathcal{P}^U_1\neq \varnothing$ and $\mathcal{P}^U_2= \varnothing$;
if $U=\langle ((12),\rho_0,\rho_0)\rangle$, then $\mathcal{P}^U_1\neq \varnothing$ and $\mathcal{P}^U_2\neq \varnothing$;
if $U=\langle (id,\rho_0,\rho_0)\rangle$ and $n=2$, then $\mathcal{P}^U_1= \varnothing$ and $\mathcal{P}^U_2\neq \varnothing$.

\vspace{8mm}

\noindent {\Large{\bf References}}

\vspace{2mm}

\noindent Bubboloni, D., Gori, M., 2014. Anonymous and neutral majority rules.
Social Choice and Welfare 43, 377-401.
\vspace{2mm}

\noindent Bubboloni, D., Gori, M., 2015. Symmetric majority rules. Mathematical Social Sciences 76, 73-86.
\vspace{2mm}

\noindent Bubboloni, D., Gori, M., 2016. On the reversal bias of the Minimax social choice correspondence.
Mathematical Social Sciences 81, 53-61.
\vspace{2mm}

\noindent Campbell, D.E., Kelly, J.S., 2011. Majority selection of one alternative from a binary agenda.
Economics Letters 110, 272-273.
\vspace{2mm}

\noindent Campbell, D.E., Kelly, J.S., 2013. Anonymity, monotonicity, and limited neutrality: selecting a single alternative
from a binary agenda. Economics Letters 118, 10-12.
\vspace{2mm}

\noindent Campbell, D.E., Kelly, J.S., 2015. The finer structure of resolute, neutral, and anonymous social choice correspondences. Economics Letters 132, 109-111.
\vspace{2mm}

\noindent Do\u gan, O., Giritligil, A.E., 2015. Anonymous and neutral social choice: existence results on resoluteness. Murat Sertel Center for Advanced Economic Studies, Working Paper Series 2015-01.
\vspace{2mm}

\noindent Fishburn, P.C., 1977. Condorcet social choice functions. SIAM Journal on Applied Mathematics 33, 469-489.
\vspace{2mm}

\noindent Jacobson, N., 1974. {\it Basic Algebra I }. W.H. Freeman and Company, New York.
\vspace{2mm}

\noindent Moulin, H., 1983. \textit{The strategy of social choice}. Advanced Textbooks in Economics, North Holland Publishing Company.
\vspace{2mm}

\noindent Moulin, H., 1988. Condorcet's principle implies the no-show paradox. Journal of Economic Theory 45, 53-64.
\vspace{2mm}

\noindent Saari, D.G., 1994.  {\it Geometry of voting}.  Studies in Economic Theory, Volume 3. Springer.
\vspace{2mm}

\noindent Saari, D.G., Barney, S., 2003. Consequences of reversing preferences. The Mathematical Intelligencer 25, 17-31.
\vspace{2mm}

\end{document}